\documentclass[draftclsnofoot,onecolumn,11pt]{IEEEtran}

\usepackage{amsmath}
\usepackage{amssymb}
\usepackage{amsthm}
\usepackage{dsfont}
\usepackage{graphicx}
\usepackage{cases,cite}
\usepackage{algpseudocode}
\usepackage{algorithm}

\newtheorem{thm} {Theorem}
\newtheorem{lem} {Lemma}

\newtheorem{assumption} {Assumption}

\begin{document}
\title{Distributed Decision-Making over Adaptive Networks}

\author{Sheng-Yuan Tu, \IEEEmembership{Student Member,~IEEE}
        and Ali H. Sayed, \IEEEmembership{Fellow,~IEEE}
\thanks{Copyright (c) 2013 IEEE. Personal use of this material is permitted.
However, permission to use this material for any other purposes must
be obtained from the IEEE by sending a request to
pubs-permissions@ieee.org.

This work was supported in part by NSF grant CCF-1011918. An earlier
conference version of this work appeared in \cite{Tu12c}. The
authors are with the Department of Electrical Engineering,
University of California, Los Angeles (e-mail: shinetu@ee.ucla.edu;
sayed@ee.ucla.edu).}}

\maketitle

\begin{abstract}
In distributed processing, agents generally collect data generated
by the \emph{same} underlying unknown model (represented by a vector
of parameters) and then solve an estimation or inference task
cooperatively. In this paper, we consider the situation in which the
data observed by the agents may have risen from two \emph{different}
models. Agents do not know beforehand which model accounts for their
data and the data of their neighbors. The objective for the network
is for all agents to reach agreement on which model to track and to
estimate this model cooperatively. In these situations, where agents
are subject to data from unknown different sources, conventional
distributed estimation strategies would lead to biased estimates
relative to any of the underlying models. We first show how to
modify existing strategies to guarantee unbiasedness. We then
develop a classification scheme for the agents to identify the
models that generated the data, and propose a procedure by which the
entire network can be made to converge towards the same model
through a collaborative decision-making process. The resulting
algorithm is applied to model fish foraging behavior in the presence
of two food sources.
\end{abstract}


\section{Introduction}
Self-organization is a remarkable property of biological networks
\cite{Camazine03,Couzin09}, where various forms of complex behavior
are evident and result from decentralized interactions among agents
with limited capabilities. One example of sophisticated behavior is
the group decision-making process by animals \cite{Sumpter09}. For
example, it is common for biological networks to encounter
situations where agents need to decide between multiple options,
such as fish deciding between following one food source or another
\cite{Couzin11}, and bees or ants deciding between moving towards a
new hive or another \cite{Britton02,Pratt02}. Although multiple
options may be available, the agents are still able to reach
agreement in a decentralized manner and move towards a common
destination (e.g., \cite{Beekman06}).

In previous works, we proposed and studied several diffusion
strategies
\cite{Lopes08,Cattivelli08,Cattivelli10,Chen12,Sayed13a,Sayed13b}
that allow agents to adapt and learn through a process of in-network
collaboration and learning. References \cite{Sayed13a,Sayed13b}
provide overviews of diffusion techniques and their application to
distributed adaptation, learning, and optimization over networks.
Examples of further applications and studies appear, e.g., in
\cite{Li09,Chouvardas11,Lorenzo13b,Xia11b,Takahashi10b,Chouvardas12}.
Diffusion networks consist of a collection of adaptive agents that
are able to respond to excitations in real-time. Compared with the
class of consensus strategies
\cite{Tsitsiklis86,Nedic09,Schizas09,Mateos09,Dimakis10,Kar11,Kar12},
diffusion networks have been shown to remain stable irrespective of
the network topology, while consensus networks can become unstable
even when each agent is individually stable \cite{Tu12a}. Diffusion
strategies have also been shown to lead to improved convergence rate
and superior mean-square-error performance \cite{Tu12a,Sayed13b}.
For these reasons, we focus in the remainder of this paper on the
use of diffusion strategies for decentralized decision-making.

Motivated by the behavior of biological networks, we study
distributed decision-making over networks where agents are subject
to data arising from two different models. The agents do not know
beforehand which model accounts for their data and the data of their
neighbors. The objective of the network is for all agents to reach
agreement on one model and to estimate and track this \emph{common}
model cooperatively. The task of reaching agreement over a network
of agents subjected to different models is more challenging than
earlier works on inference under a single data model. The difficulty
is due to various reasons. First, traditional (consensus and
diffusion) strategies will converge to a biased solution (see Eq.
(\ref{eq91})). We therefore need a mechanism to compensate for the
bias. Second, each agent now needs to distinguish between which
model each of its neighbors is collecting data from (this is called
the \emph{observed} model) and which model the network is evolving
to (this is called the \emph{desired} model). In other words, in
addition to the learning and adaptation process for tracking, the
agents should be equipped with a classification scheme to
distinguish between the observed and desired models. The agents also
need to be endowed with a decision process to agree among themselves
on a common (desired) model to track. Moreover, the classification
scheme and the decision-making process will need to be implemented
in a fully distributed manner and in real-time, alongside the
adaptation process.

There have been useful prior works in the literature on formations
over multi-agent networks
\cite{Jadbabaie03,Olfati04,Fax04,Olfati06,Yildiz10,Khan10,Forero10}
and opinion formation over social networks
\cite{Castellano09,Acemoglu10,Jadbabaie12} using, for example,
consensus strategies. These earlier works are mainly interested in
having the agents reach an average consensus state, whereas in our
problem formulation agents will need to reach one of the models and
not the average of both models. Another difference between this work
and the earlier efforts is our focus on combining real-time
classification, decision-making, and adaptation into a single
integrated framework running at each agent. To do so, we need to
show how the distributed strategy should be modified to remove the
bias that would arise due to the multiplicity of models
--- without this step, the combined decision-making and adaptation scheme
will not perform as required. In addition, in our formulation, the
agents need to continuously adjust their decisions and their
estimates because the models are allowed to change over time. In
this way, reaching a static consensus is not the objective of the
network. Instead, the agents need to continuously adjust and track
in a dynamic environment where decisions and estimates evolve with
time as necessary. Diffusion strategies endow networks with such
tracking abilities --- see, e.g., Sec. VII of \cite{Zhao12}, where
it is shown how well these strategies track as a function of the
level of non-stationarity in the underlying models.


\section{Diffusion Strategy}
Consider a collection of $N$ agents (or nodes) distributed over a
geographic region. The set of neighbors (i.e. neighborhood) of node
$k$ is denoted by $\mathcal{N}_k$; the number of nodes in
$\mathcal{N}_k$ is denoted by $n_k$. At every time instant, $i$,
each node $k$ is able to observe realizations $\{d_k(i),u_{k,i}\}$
of a scalar random process $\boldsymbol{d}_{k}(i)$ and a $1\times M$
\emph{row} random regressor $\boldsymbol{u}_{k,i}$ with a
positive-definite covariance matrix,
$R_{u,k}=\mathbb{E}\boldsymbol{u}_{k,i}^T\boldsymbol{u}_{k,i}>0$.
The regressors $\{\boldsymbol{u}_{k,i}\}$ are assumed to be
temporally white and spatially independent, i.e.,
$\mathbb{E}\boldsymbol{u}^T_{k,i}\boldsymbol{u}_{l,j}=
R_{u,k}\delta_{kl}\delta_{ij}$ in terms of the Kronecker delta
function. Note that we are denoting random quantities by boldface
letters and their realizations or deterministic quantities by normal
letters. The data $\{\boldsymbol{d}_{k}(i),\boldsymbol{u}_{k,i}\}$
collected at node $k$ are assumed to originate from one of two
unknown \emph{column} vectors $\{w^\circ_0,w^\circ_1\}$ of size $M$
in the following manner. We denote the generic observed model by
$z_k^\circ\in \{w_0^\circ, w_1^\circ\}$; node $k$ does not know
beforehand the observed model. The data at node $k$ are related to
its observed model $z^\circ_k$ via a linear regression model of the
form:
\begin{equation} \label{eq2}
    \boldsymbol{d}_k(i) = \boldsymbol{u}_{k,i}z^\circ_k+\boldsymbol{v}_k(i)
\end{equation}
where $\boldsymbol{v}_k(i)$ is measurement noise with variance
$\sigma^2_{v,k}$ and assumed to be temporally white and spatially
independent. The noise $\boldsymbol{v}_k(i)$ is assumed to be
independent of $\boldsymbol{u}_{l,j}$ for all $\{k,l,i,j\}$. All
random processes are zero mean.

The objective of the network is to have \emph{all} agents converge
to an estimate for \emph{one} of the models. For example, if the
models happen to represent the location of food sources
\cite{Tu11a,Sayed13a}, then this agreement will make the agents move
towards one particular food source in lieu of the other source. More
specifically, let $\boldsymbol{w}_{k,i}$ denote the estimator for
$z^\circ_k$ at node $k$ at time $i$. The network would like to reach
an ageement on a common $q$, such that
\begin{equation}\label{eq23}
    \boldsymbol{w}_{k,i} \rightarrow w^\circ_q
    \text{ for $q=0$ or $q=1$ and for all $k$ as $i\rightarrow\infty$}
\end{equation}
where convergence is in some desirable sense (such as the
mean-square-error sense).

Several \emph{adaptive} diffusion strategies for distributed
estimation under a common model scenario were proposed and studied
in \cite{Lopes08,Cattivelli10,Cattivelli08,Chen12,Sayed13a},
following the developments in
\cite{Lopes06,Sayed07,Lopes07,Cattivelli07,Cattivelli08b} ---
overviews of these results appear in \cite{Sayed13a,Sayed13b}. One
such scheme is the adaptive-then-combine (ATC) diffusion strategy
\cite{Cattivelli08b,Cattivelli10}. It operates as follows. We select
an $N\times N$ matrix $A$ with nonnegative entries $\{a_{l,k}\}$
satisfying:
\begin{equation} \label{eq50}
    \mathds{1}^T_NA=\mathds{1}^T_N \quad\text{and}\quad
    a_{l,k}=0 \text{ if $l\notin\mathcal{N}_k$}
\end{equation}
where $\mathds{1}_N$ is the vector of size $N$ with all entries
equal to one. The entry $a_{l,k}$ denotes the weight that node $k$
assigns to data arriving from node $l$ (see Fig. \ref{Fig_1}). The
ATC diffusion strategy updates $\boldsymbol{w}_{k,i-1}$ to
$\boldsymbol{w}_{k,i}$ as follows:
\begin{align}\label{eq49}
    \boldsymbol{\psi}_{k,i} &= \boldsymbol{w}_{k,i-1}+
    \mu_k\cdot \boldsymbol{u}_{k,i}^T[\boldsymbol{d}_k(i)-
    \boldsymbol{u}_{k,i}\boldsymbol{w}_{k,i-1}]\\
    \boldsymbol{w}_{k,i} &= \sum_{l\in\mathcal{N}_k}a_{l,k}\boldsymbol{\psi}_{l,i} \label{eq1}
\end{align}
where $\mu_k$ is the \emph{constant} positive step-size used by node
$k$. The first step (\ref{eq49}) involves local adaptation, where
node $k$ uses its own data
$\{\boldsymbol{d}_k(i),\boldsymbol{u}_{k,i}\}$ to update the weight
estimate at node $k$ from $\boldsymbol{w}_{k,i-1}$ to an
intermediate value $\boldsymbol{\psi}_{k,i}$. The second step
(\ref{eq1}) is a combination step where the intermediate estimates
$\{\boldsymbol{\psi}_{l,i}\}$ from the neighborhood of node $k$ are
combined through the weights $\{a_{l,k}\}$ to obtain the updated
weight estimate $\boldsymbol{w}_{k,i}$. Such diffusion strategies
have found applications in several domains including distributed
optimization, adaptation, learning, and the modeling of biological
networks --- see, e.g., \cite{Sayed13a,Sayed13b,Tu11a} and the
references therein. Diffusion strategies were also used in some
recent works \cite{Ram10,Bianchi11,Srivastava11,Stankovic11} albeit
with diminishing step-sizes ($\mu_k(i)\rightarrow 0$) to enforce
consensus among nodes. However, decaying step-sizes disable
adaptation once they approach zero. Constant step-sizes are used in
(\ref{eq49})-(\ref{eq1}) to enable continuous adaptation and
learning, which is critical for the application under study in this
work.

\begin{figure}
\centering
\includegraphics[width=14em]{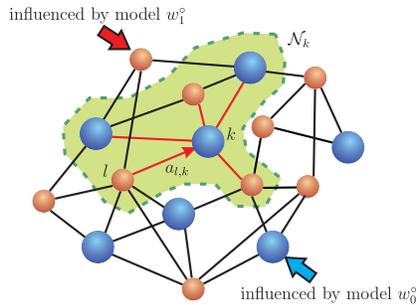}
\caption{A connected network where data collected by the agents are
influenced by one of two models. The weight $a_{l,k}$ scales the
data transmitted from node $l$ to node $k$ over the edge linking
them.} \label{Fig_1}
\end{figure}

When the data arriving at the nodes could have risen from one model
or another, the distributed strategy (\ref{eq49})-(\ref{eq1}) will
not be able to achieve agreement as in (\ref{eq23}) and the
resulting weight estimates will tend towards a biased value. We
first explain how this degradation arises and subsequently explain
how it can be remedied.
\begin{assumption}[Strongly connected network]\label{ass3}
The network topology is assumed to be strongly connected so that the
corresponding combination matrix $A$ is primitive, i.e., there
exists an integer power $j>0$ such that $[A^j]_{l,k} > 0$ for all
$l$ and $k$.
\end{assumption}
\noindent As explained in \cite{Sayed13a}, Assumption \ref{ass3}
amounts to requiring the network to be connected (where a path with
nonzero weights exists between any two nodes), and for at least one
node to have a non-trivial self-loop (i.e., $a_{k,k}>0$ for at least
one $k$). We conclude from the Perron-Frobenius Theorem
\cite{Horn85,Berman94} that every primitive left-stochastic matrix
$A$ has a unique eigenvalue at one while all other eigenvalues are
strictly less than one in magnitude. Moreover, if we denote the
right-eigenvector that is associated with the eigenvalue at one by
$c$ and normalize its entries to add up to one then it holds that:
\begin{equation}\label{eq9}
    Ac=c,\quad \mathds{1}_N^T c=1,\quad
    \text{and}\quad0<c_k<1.
\end{equation}

Let us assume for the time being that the agents in the network have
agreed on converging towards one of the models (but they do not know
beforehand which model it will be). We denote the desired model
generically by $w_q^\circ$. In Section \ref{sec2}, we explain how
this agreement process can be attained. Here we explain that even
when agreement is present, the diffusion strategy
(\ref{eq49})-(\ref{eq1}) leads to biased estimates unless it is
modified in a proper way. To see this, we introduce the following
error vectors for any node $k$:
\begin{equation}\label{eq29}
    \tilde{\boldsymbol{w}}_{k,i}\triangleq
    w^\circ_q-\boldsymbol{w}_{k,i}\;\;\text{ and }\;\;
    \tilde{z}^\circ_{k}\triangleq w^\circ_q-z^\circ_k.
\end{equation}
Then, using model (\ref{eq2}), we obtain that the update vector in
(\ref{eq49}) becomes
\begin{align}
    \boldsymbol{h}_{k,i}&\triangleq
    \boldsymbol{u}^T_{k,i}[\boldsymbol{d}_k(i)-
    \boldsymbol{u}_{k,i}\boldsymbol{w}_{k,i-1}]\notag \\
    &=\boldsymbol{u}^T_{k,i}\boldsymbol{u}_{k,i}\tilde{\boldsymbol{w}}_{k,i-1}
    -\boldsymbol{u}^T_{k,i}\boldsymbol{u}_{k,i}\tilde{z}^\circ_{k}
    +\boldsymbol{u}^T_{k,i}\boldsymbol{v}_k(i).\label{eq20}
\end{align}
We collect all error vectors across the network into block vectors:
$\tilde{\boldsymbol{w}}_{i}\triangleq
\text{col}\left\{\tilde{\boldsymbol{w}}_{k,i}\right\}$ and
$\tilde{z}^\circ\triangleq
\text{col}\left\{\tilde{z}^\circ_{k}\right\}$. We also collect the
step-sizes into a block diagonal matrix and introduce the extended
combination matrix:
\begin{equation}\label{eq56}
    \mathcal{M} = \text{diag}\{\mu_kI_M\}\quad \text{and}\quad
    \mathcal{A}\triangleq A\otimes I_M
\end{equation}
where $I_M$ denotes the identity matrix of size $M$. In
(\ref{eq56}), the notation $\text{diag}\{\cdot\}$ constructs a
diagonal matrix from its arguments and the symbol $\otimes$ denotes
the Kronecker product of two matrices. Moreover, the notation
$\text{col}\{\cdot\}$ denotes the vector that is obtained by
stacking its arguments on top of each other. Then, starting from
(\ref{eq49})-(\ref{eq1}) and using relation (\ref{eq20}), we can
verify that the global error vector $\tilde{\boldsymbol{w}}_{i}$ of
the network evolves over time according to the recursion:
\begin{equation}\label{eq3}
    \tilde{\boldsymbol{w}}_{i}=
    \boldsymbol{\mathcal{B}}_i \tilde{\boldsymbol{w}}_{i-1}+
    \boldsymbol{y}_i
\end{equation}
where the matrix $\boldsymbol{\mathcal{B}}_i$ and the vector
$\boldsymbol{y}_i$ are defined in Table \ref{tab1} with
$\boldsymbol{\mathcal{R}}_i\triangleq\text{diag}\{
\boldsymbol{u}^T_{k,i}\boldsymbol{u}_{k,i}\}$ and
$\boldsymbol{s}_i\triangleq\text{col}\{
\boldsymbol{u}^T_{k,i}\boldsymbol{v}_{k,i}\}$. Note that the matrix
$\boldsymbol{\mathcal{B}}_i$ is a random matrix due to the
randomness of the regressors $\{\boldsymbol{u}_{k,i}\}$. Since the
regressors are temporally white and spatially independent, then
$\boldsymbol{\mathcal{B}}_{i}$ is independent of
$\tilde{\boldsymbol{w}}_{i-1}$. In addition, since
$\boldsymbol{u}_{k,i}$ is independent of $\boldsymbol{v}_{k}(i)$,
the vector $\boldsymbol{s}_i$ in $\boldsymbol{y}_i$ has zero mean.
Then, from (\ref{eq3}), the mean of $\tilde{\boldsymbol{w}}_{i}$
evolves over time according to the recursion:
\begin{equation}\label{eq5}
    \mathbb{E}\tilde{\boldsymbol{w}}_{i}=\mathcal{B}\cdot
    \mathbb{E}\tilde{\boldsymbol{w}}_{i-1}+y
\end{equation}
where $\mathcal{B}\triangleq \mathbb{E}\boldsymbol{\mathcal{B}}_{i}$
and $y \triangleq \mathbb{E}\boldsymbol{y}_{i}$ are defined in Table
\ref{tab1} with $\mathcal{R}\triangleq
\mathbb{E}\boldsymbol{\mathcal{R}}_{i} =\text{diag}\{R_{u,k}\}$. It
can be easily verified that a necessary and sufficient condition to
ensure the convergence of $\mathbb{E}\tilde{\boldsymbol{w}}_{i}$ in
(\ref{eq5}) to zero is
\begin{equation}\label{eq22}
    \rho(\mathcal{B})<1\quad\text{and}\quad
    y=0
\end{equation}
where $\rho(\cdot)$ denotes the spectral radius of its argument. It
was verified in \cite{Sayed13a,Tu12a} that a sufficient condition to
ensure $\rho({\cal B})<1$ is to select the site-sizes $\{\mu_k\}$
such that
\begin{equation}\label{eq26}
    0 < \mu_k < \frac{2}{\rho(R_{u,k})}
\end{equation}
for all $k$. This conclusion is independent of $A$. However, for the
second condition in (\ref{eq22}), we note that in general, the
vector $y=\mathcal{A}^T\mathcal{M}\mathcal{R}\tilde{z}^\circ$ cannot
be zero no matter how the nodes select the combination matrix $A$.
When this happens, the weight estimate will be biased. Let us
consider the example with three nodes in Fig. \ref{Fig_2} where node
1 observes data from model $w^\circ_0$, while nodes 2 and 3 observe
data from another model $w^\circ_1$. The matrix $A$ in this case is
shown in Fig. \ref{Fig_2} with the parameters $\{a,b,c,d\}$ lying in
the interval $[0,1]$ and $b+c\leq1$. We assume that the step-sizes
and regression covariance matrices are the same, i.e., $\mu_k=\mu$
and $R_{u,k}=R_u$ for all $k$. If the desired model of the network
is $w^\circ_q=w^\circ_0$, then the third block of $y$ becomes $\mu
R_u(w^\circ_0-w^\circ_1)$, which can never become zero no matter
what the parameters $\{a,b,c,d\}$ are. More generally, using results
on the limiting behavior of the estimation errors
$\{\tilde{\boldsymbol{w}}_{k,i}\}$ from \cite{Chen13b}, we can
characterize the limiting point of the diffusion strategy
(\ref{eq49})-(\ref{eq1}) as follows.

\begin{table}
\centering \caption{{\rm The error vector evolves according to the
recursion} $\tilde{\boldsymbol{w}}_{i}=\boldsymbol{\mathcal{B}}_i
    \tilde{\boldsymbol{w}}_{i-1}+\boldsymbol{y}_i$,
{\rm where the variables} $\{\boldsymbol{\mathcal{B}}_i,
\boldsymbol{y}_i\}$ {\rm and their respective means are listed below
for the conventional and modified diffusion strategies.} }
\renewcommand{\arraystretch}{1.5}\label{tab1}
\begin{tabular}{|c|c|c|}
\hline & \textbf{Diffusion} (\ref{eq49})-(\ref{eq1}) &
\textbf{Modified diffusion} (\ref{eq6})-(\ref{eq15})  \\
\hline\hline $\boldsymbol{\mathcal{B}}_i$ &
    $\mathcal{A}^{T}(I_{NM}-\mathcal{M}\boldsymbol{\mathcal{R}}_i)$ &
    $\mathcal{A}_1^{T}(I_{NM}-\mathcal{M}\boldsymbol{\mathcal{R}}_i)
    +\mathcal{A}_2^{T}$ \\
\hline $\mathcal{B}\triangleq\mathbb{E}\boldsymbol{\mathcal{B}}_{i}$
&   $\mathcal{A}^T(I_{NM}-\mathcal{M}\mathcal{R})$ &
    $\mathcal{A}_1^T(I_{NM}-\mathcal{M}\mathcal{R})+\mathcal{A}_2^T$ \\
\hline $\boldsymbol{y}_i$ &
    $\mathcal{A}^T\mathcal{M}\boldsymbol{\mathcal{R}}_i\tilde{z}^\circ-
    \mathcal{A}^T\mathcal{M}\boldsymbol{s}_i$ &
    $\mathcal{A}_1^T\mathcal{M}\boldsymbol{\mathcal{R}}_i\tilde{z}^\circ-
    \mathcal{A}_1^T\mathcal{M}\boldsymbol{s}_i$ \\
\hline $y\triangleq \mathbb{E}\boldsymbol{y}_i$ &
    $\mathcal{A}^T\mathcal{M}\mathcal{R}\tilde{z}^\circ$ &
    $\mathcal{A}_1^T\mathcal{M}\mathcal{R}\tilde{z}^\circ$ \\
\hline
\end{tabular}
\end{table}

\begin{figure}
\centering
\includegraphics[width=20em]{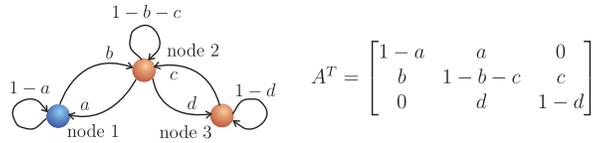}
\caption{A three-node network. Node 1 observes data from $w^\circ_0$
while nodes 2 and 3 observe data from $w^\circ_1$.} \label{Fig_2}
\end{figure}

\begin{lem}\label{lem3}
For the diffusion strategy (\ref{eq49})-(\ref{eq1}) with $\mu_k=\mu$
and $R_{u,k}=R_u$ for all $k$ and for sufficiently small step-sizes,
all weight estimators $\{\boldsymbol{w}_{k,i}\}$ converge to a limit
point $w^\circ$ in the mean-square sense, i.e.,
$\mathbb{E}\|w^\circ-\boldsymbol{w}_{k,i}\|^2$ is bounded and of the
order of $\mu$, where $w^\circ$ is given by
\begin{equation}\label{eq91}
    w^\circ = \sum_{k=1}^Nc_kz^\circ_k
\end{equation}
where the vector $c$ is defined in (\ref{eq9}).
\end{lem}{
\begin{proof}
The result follows from Eq. (25) in \cite{Chen13b} by noting that
the variable $s_k(w^\circ)$ used in \cite{Chen13b} is given by
$R_{u}(z^\circ_k-w^\circ)$.
\end{proof}}
\noindent Thus, when the agents collect data from different models,
the estimates using the diffusion strategy (\ref{eq49})-(\ref{eq1})
converge to a convex combination of these models given by
(\ref{eq91}), which is different from any of the individual models
because $c_k>0$ for all $k$. A similar conclusion holds for the case
of non-uniform step-sizes $\{\mu_k\}$ and covariance matrices
$\{R_{u,k}\}$.

\section{Modified Diffusion Strategy}
To deal with the problem of bias, we now show how to modify the
diffusion strategy (\ref{eq49})-(\ref{eq1}). We observe from the
example in Fig. \ref{Fig_2} that the third entry of the vector $y$
cannot be zero because the neighbor of node $3$ observes data
arising from a model that is different from the desired model. Note
from (\ref{eq20}) that the bias term arises from the gradient
direction used in computing the intermediate estimates in
(\ref{eq49}). These observations suggest that to ensure unbiased
mean convergence, a node should not combine intermediate estimates
from neighbors whose observed model is different from the desired
model. For this reason, we shall replace the intermediate estimates
from these neighbors by their previous estimates
$\{\boldsymbol{w}_{l,i-1}\}$ in the combination step (\ref{eq1}).
Specifically, we shall adjust the diffusion strategy
(\ref{eq49})-(\ref{eq1}) as follows:
\begin{align}\label{eq6}
    \boldsymbol{\psi}_{k,i} &= \boldsymbol{w}_{k,i-1}+
    \mu_k\cdot\boldsymbol{u}_{k,i}^T
    [\boldsymbol{d}_k(i)-\boldsymbol{u}_{k,i}\boldsymbol{w}_{k,i-1}]\\
    \boldsymbol{w}_{k,i} &= \sum_{l\in\mathcal{N}_k}\left(a^{(1)}_{l,k}
    \boldsymbol{\psi}_{l,i}
    +a^{(2)}_{l,k}\boldsymbol{w}_{l,i-1}\right)\label{eq15}
\end{align}
where the $\{a^{(1)}_{l,k}\}$ and $\{a^{(2)}_{l,k}\}$ are two sets
of nonnegative scalars and their respective combination matrices
$A_1$ and $A_2$ satisfy
\begin{equation}\label{eq111}
    A_1+A_2=A
\end{equation}
with $A$ being the original left-stochastic matrix in (\ref{eq50}).
Note that step (\ref{eq6}) is the same as step (\ref{eq49}).
However, in the second step (\ref{eq15}), nodes aggregate the
$\{\boldsymbol{\psi}_{l,i},\boldsymbol{w}_{l,i-1}\}$ from their
neighborhood. With such adjustment, we will verify that by properly
selecting $\{a^{(1)}_{l,k},a^{(2)}_{l,k}\}$, unbiased mean
convergence can be guaranteed. The choice of which entries of $A$ go
into $A_1$ or $A_2$ will depend on which of the neighbors of node
$k$ are observing data arising from a model that agrees with the
desired model for node $k$.

\subsection{Construction of Matrices $A_1$ and $A_2$}
To construct the matrices $\{A_1,A_2\}$ we associate two vectors
with the network, $f$ and $g_i$. Both vectors are of size $N$. The
vector $f$ is fixed and its $k$th entry, $f(k)$, is set to $f(k)=0$
when the observed model for node $k$ is $w_0^\circ$; otherwise, it
is set to $f(k)=1$. On the other hand, the vector $g_i$ is evolving
with time; its $k$th entry is set to $g_i(k)=0$ when the desired
model for node $k$ is $w_0^\circ$; otherwise, it is set equal to
$g_i(k)=1$. Then, we shall set the entries of $A_1$ and $A_2$
according to the following rules:
\begin{align}\label{eq24}
    & a^{(1)}_{l,k,i} =
    \begin{cases}
    a_{l,k}, &\text{if $l\in\mathcal{N}_k$ and $f(l)= g_i(k)$}\\
    0, &\text{otherwise}
    \end{cases}\\
    & a^{(2)}_{l,k,i} =
    \begin{cases}
    a_{l,k}, &\text{if $l\in\mathcal{N}_k$ and $f(l)\neq g_i(k)$}\\
    0, &\text{otherwise}
    \end{cases}.\label{eq25}
\end{align}
That is, nodes that observe data arising from the same model that
node $k$ wishes to converge to will be reinforced and their
intermediate estimates $\{\boldsymbol{\psi}_{l,i}\}$ will be used
(their combination weights are collected into matrix $A_1$). On the
other hand, nodes that observe data arising from a different model
than the objective for node $k$ will be de-emphasized and their
prior estimates $\{\boldsymbol{w}_{l,i-1}\}$ will be used in the
combination step (\ref{eq15}) (their combination weights are
collected into matrix $A_2$). Note that the scalars
$\{a^{(1)}_{l,k,i},a^{(2)}_{l,k,i}\}$ in (\ref{eq24})-(\ref{eq25})
are now indexed with time due to their dependence on $g_i(k)$.

\subsection{Mean-Error Analysis}
It is important to note that to construct the combination weights
from (\ref{eq24})-(\ref{eq25}), each node $k$ needs to know what are
the observed models influencing its neighbors (i.e., $f(l)$ for
$l\in\mathcal{N}_k$); it also needs to know how to update its
objective in $g_i(k)$ so that the $\{g_i(l)\}$ converge to the same
value. In the next two sections, we will describe a distributed
decision-making procedure by which the nodes are able to achieve
agreement on $\{g_i(k)\}$. We will also develop a classification
scheme to estimate $\{f(l)\}$ using available data. More
importantly, the convergence of the vectors $\{f,g_i\}$ will occur
before the convergence of the adaptation process to estimate the
agreed-upon model. Therefore, let us assume for the time being that
the nodes know the $\{f(l)\}$ of their neighbors and have achieved
agreement on the desired model, which we are denoting by
$w_q^\circ$, so that (see Eq. (\ref{eq127}) in Theorem \ref{thm1})
\begin{equation}\label{eq30}
    g_i(1)=g_i(2)=\cdots=g_i(N)=q,\;\; \text{ for all $i$.}
\end{equation}

Using relation (\ref{eq20}) and the modified diffusion strategy
(\ref{eq6})-(\ref{eq15}), the recursion for the global error vector
$\tilde{\boldsymbol{w}}_{i}$ is again given by (\ref{eq3}) with the
matrix $\boldsymbol{\mathcal{B}}_i$ and the vector
$\boldsymbol{y}_i$ defined in Table \ref{tab1} and the combination
matrices $\mathcal{A}_1$ and $\mathcal{A}_2$ defined in a manner
similar to $\mathcal{A}$ in (\ref{eq56}). We therefore get the same
mean recursion as (\ref{eq5}) with the matrix $\mathcal{B}$ and the
vector $y$ defined in Table \ref{tab1}. The following result
establishes asymptotic mean convergence for the modified diffusion
strategy (\ref{eq6})-(\ref{eq15}).
\begin{thm}\label{thm2}
Under condition (\ref{eq30}), the modified diffusion strategy
(\ref{eq6})-(\ref{eq15}) converges in the mean if the matrices $A_1$
and $A_2$ are constructed according to (\ref{eq24})-(\ref{eq25}) and
the step-sizes $\{\mu_k\}$ satisfy condition (\ref{eq26}) for those
nodes whose observed model is the same as the desired model
$w^\circ_q$ for the network.
\end{thm}
\begin{proof}
See Appendix \ref{appC}.
\end{proof}
We conclude from the argument in Appendix \ref{appC} that the net
effect of the construction (\ref{eq24})--(\ref{eq25}) is the
following. Let $w_q^\circ$ denote the desired model that the network
wishes to converge to. We denote by $\mathcal{N}_q$ the subset of
nodes that receive data arising from the \emph{same} model. The
remaining nodes belong to the set $\mathcal{N}^c_q$. Nodes that
belong to the set $\mathcal{N}_q$ run the traditional diffusion
strategy (\ref{eq49})-(\ref{eq1}) using the combination matrix $A$
and their step-sizes are required to satisfy (\ref{eq26}). The
remaining nodes in $\mathcal{N}_q^c$ set their step-sizes to zero
and run only step (\ref{eq1}) of the diffusion strategy. These nodes
do not perform the adaptive update (\ref{eq49}) and therefore their
estimates satisfy $\boldsymbol{\psi}_{k,i}=\boldsymbol{w}_{k,i-1}$
for all $k\in \mathcal{N}^c_q$.

\begin{figure}
\centering
\includegraphics[width=25em]{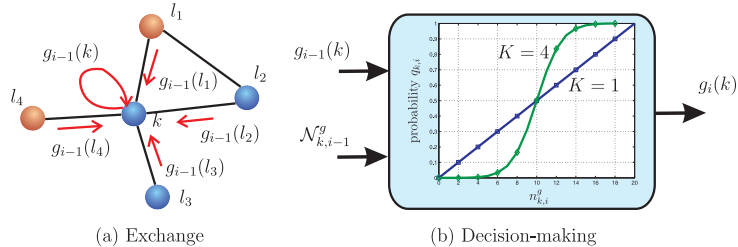}
\caption{Decision-making process (a) node $k$ receives the desired
models from its neighbors (b) node $k$ updates its desired model
using (\ref{eq17})-(\ref{eq57}) with 20 neighbors.} \label{Fig_13}
\end{figure}

\section{Distributed Decision-Making}\label{sec2}
The decision-making process is motivated by the process used by
animal groups to reach agreement, and which is known as quorum
response \cite{Britton02,Pratt02,Sumpter09}. The procedure is
illustrated in Fig. \ref{Fig_13} and described as follows. At time
$i$, every node $k$ has its previous desired model
$\boldsymbol{g}_{i-1}(k)$, now modeled as a random variable since it
will be constructed from data realizations that are subject to
randomness. Node $k$ exchanges $\boldsymbol{g}_{i-1}(k)$ with its
neighbors and constructs the set
\begin{equation}\label{eq98}
    \mathcal{N}_{k,i-1}^g =
    \{l\;|\;l\in\mathcal{N}_k,\boldsymbol{g}_{i-1}(l)=\boldsymbol{g}_{i-1}(k)\}.
\end{equation}
That is, the set $\mathcal{N}_{k,i-1}^g$ contains the subset of
nodes that are in the neighborhood of $k$ and have the same desired
model as node $k$ at time $i-1$. This set changes over time. Let
$n^g_k(i-1)$ denote the number of nodes in $\mathcal{N}_{k,i-1}^g$.
Since at least one node (node $k$) belongs to
$\mathcal{N}_{k,i-1}^g$, we have that $n^g_k(i-1)\geq 1$. Then, one
way for node $k$ to participate in the quorum response is to update
its desired model $\boldsymbol{g}_i(k)$ according to the rule:
\begin{equation}\label{eq17}
    \boldsymbol{g}_i(k) = \begin{cases}
    \boldsymbol{g}_{i-1}(k), &\text{with probability $q_{k,i-1}$}\\
    1-\boldsymbol{g}_{i-1}(k), &\text{with probability $1-q_{k,i-1}$}
    \end{cases}
\end{equation}
where the probability measure is computed as:
\begin{equation}\label{eq57}
    q_{k,i-1} = \frac{[n^g_k(i-1)]^K}{[n^g_k(i-1)]^K+[n_k-n^g_k(i-1)]^K} > 0
\end{equation}
and the exponent $K$ is a positive constant (e.g., $K=4$). That is,
node $k$ determines its desired model in a probabilistic manner, and
the probability that node $k$ maintains its desired target is
proportional to the $K$th power of the number of neighbors having
the same desired model (see Fig. \ref{Fig_13}(b)). Using the above
stochastic formulation, we are able to establish agreement on the
desired model among the nodes.
\begin{thm}\label{thm1}
For a connected network starting from an arbitrary initial selection
for the desired models vector $\boldsymbol{g}_{i}$ at time $i=-1$,
and applying the update rule (\ref{eq98})-(\ref{eq57}), then all
nodes eventually achieve agreement on some desired model, i.e.,
\begin{equation}\label{eq127}
    \boldsymbol{g}_i(1)=\boldsymbol{g}_i(2)=
    \ldots=\boldsymbol{g}_i(N),\quad \text{as $i\rightarrow \infty$.}
\end{equation}
\end{thm}
\begin{proof}
See Appendix \ref{appD}.
\end{proof}

Although rule (\ref{eq98})-(\ref{eq57}) ensures agreement on the
decision vector, this construction is still not a distributed
solution for one subtle reason: nodes need to agree on which index
(0 or 1) to use to refer to either model $\{w_0^\circ,w_1^\circ\}$.
This task would in principle require the nodes to share some global
information. We circumvent this difficulty and develop a distributed
solution as follows. Moving forward, we now associate with each node
$k$ two local vectors $\{f_k,\boldsymbol{g}_{k,i}\}$; these vectors
will play the role of local estimates for the network vectors
$\{f,\boldsymbol{g}_i\}$. Each node will then assign the index value
of one to its observed model, i.e., each node $k$ sets $f_k(k)=1$.
Then, for every $l\in\mathcal{N}_k$, the entry $f_k(l)$ is set to
one if it represents the same model as the one observed by node $k$;
otherwise, $f_k(l)$ is set to zero. The question remains about how
node $k$ knows whether its neighbors have the same observed model as
its own (this is discussed in the next section). Here we comment
first on how node $k$ adjusts the entries of its vector
$\boldsymbol{g}_{k,i-1}$. Indeed, node $k$ knows its desired model
value $\boldsymbol{g}_{k,i-1}(k)$ from time $i-1$. To assign the
remaining neighborhood entries in the vector
$\boldsymbol{g}_{k,i-1}$, the nodes in the neighborhood of node $k$
first exchange their desired model indices with node $k$, that is,
they send the information $\{\boldsymbol{g}_{l,i-1}(l),\;l\in
\mathcal{N}_k\}$ to node $k$. However, since
$\boldsymbol{g}_{l,i-1}(l)$ from node $l$ is set relative to its
$f_l(l)$, node $k$ needs to set $\boldsymbol{g}_{k,i-1}(l)$ based on
the value of $f_k(l)$. Specifically, node $k$ will set
$g_{k,i-1}(l)$ according to the rule:
\begin{equation}\label{eq31}
    \boldsymbol{g}_{k,i-1}(l) =
    \begin{cases}
    \boldsymbol{g}_{l,i-1}(l), &\text{if $f_k(l)=f_k(k)$}\\
    1-\boldsymbol{g}_{l,i-1}(l), &\text{otherwise}
    \end{cases}.
\end{equation}
That is, if node $l$ has the same observed model as node $k$, then
node $k$ simply assigns the value of $\boldsymbol{g}_{l,i-1}(l)$ to
$\boldsymbol{g}_{k,i-1}(l)$.

In this way, computations that depend on the network vectors
$\{f,\boldsymbol{g}_{i}\}$ will be replaced by computations using
the local vectors $\{f_k,\boldsymbol{g}_{k,i}\}$. That is, the
quantities $\{f(l),g_i(l)\}$ in (\ref{eq24})-(\ref{eq25}) and
(\ref{eq98})-(\ref{eq57}) are now replaced by
$\{f_k(l),\boldsymbol{g}_{k,i}(l)\}$. We verify in the following
that using the network vectors $\{f,\boldsymbol{g}_{i}\}$ is
equivalent to using the local vectors
$\{f_k,\boldsymbol{g}_{k,i}\}$.
\begin{lem}
It holds that
\begin{align}\label{eq34}
    f(l) \oplus \boldsymbol{g}_i(k) &= f_k(l)\oplus \boldsymbol{g}_{k,i}(k) \\
    \boldsymbol{g}_i(l) \oplus \boldsymbol{g}_{i}(k) &=
    \boldsymbol{g}_{k,i}(l)\oplus \boldsymbol{g}_{k,i}(k)\label{eq35}
\end{align}
where the symbol $\oplus$ denotes the exclusive-OR operation.
\end{lem}{
\begin{proof}
Since the values of
$\{f_k(l),\boldsymbol{g}_{l,i}(l),\boldsymbol{g}_{k,i}(l)\}$ are set
relative to $f_k(k)$, it holds that
\begin{align}
    f(k) \oplus f(l) &= f_k(k)\oplus f_k(l) \\
    f(k) \oplus \boldsymbol{g}_i(k) &= f_k(k)\oplus \boldsymbol{g}_{k,i}(k) \\
    f(k) \oplus \boldsymbol{g}_i(l) &= f_k(k)\oplus \boldsymbol{g}_{k,i}(l)
\end{align}
Then relations (\ref{eq34}) and (\ref{eq35}) hold in view of the
fact:
\begin{equation}
    (a \oplus b) \oplus (a \oplus e) = b \oplus e
\end{equation}
for any $a$, $b$, and $e\in\{0,1\}$.
\end{proof}}
\noindent With these replacements, node $k$ still needs to set the
entries $\{f_k(l)\}$ that correspond to its neighbors, i.e., it
needs to differentiate between their underlying models and whether
their data arise from the same model as node $k$ or not. We propose
next a procedure to determine $f_k$ at node $k$ using the available
estimates $\{\boldsymbol{w}_{l,i-1},\boldsymbol{\psi}_{l,i}\}$ for
$l\in\mathcal{N}_k$.

\section{Model Classification Scheme}\label{sec1}
To determine the vector $f_k$, we introduce the belief vector
$\boldsymbol{b}_{k,i}$, whose $l$th entry,
$\boldsymbol{b}_{k,i}(l)$, will be a measure of the belief by node
$k$ that node $l$ has the same observed model. The value of
$\boldsymbol{b}_{k,i}(l)$ lies in the range $[0,1]$. The higher the
value of $\boldsymbol{b}_{k,i}(l)$ is, the more confidence node $k$
has that node $l$ is subject to the same model as its own model. In
the proposed construction, the vector $\boldsymbol{b}_{k,i}$ will be
changing over time according to the estimates
$\{\boldsymbol{w}_{l,i-1},\boldsymbol{\psi}_{l,i}\}$. Node $k$ will
be adjusting $\boldsymbol{b}_{k,i}(l)$ according to the rule:
\begin{equation}\label{eq32}
    \boldsymbol{b}_{k,i}(l)=
    \begin{cases}
    \alpha \boldsymbol{b}_{k,i-1}(l) + (1-\alpha), &\text{to increase belief}\\
    \alpha \boldsymbol{b}_{k,i-1}(l), &\text{to decrease belief}
    \end{cases}
\end{equation}
for some positive scalar $\alpha\in(0,1)$, e.g., $\alpha=0.95$. That
is, node $k$ increases the belief by combining in a convex manner
the previous belief with the value one. Node $k$ then estimates
$f_k(l)$ according to the rule:
\begin{equation}\label{eq37}
    \hat{\boldsymbol{f}}_{k,i}(l) =
    \begin{cases}
    1, &\text{if $\boldsymbol{b}_{k,i}(l)\geq 0.5$}\\
    0, &\text{otherwise}
    \end{cases}
\end{equation}
where $\hat{\boldsymbol{f}}_{k,i}(l)$ denotes the estimate for
$f_k(l)$ at time $i$ and is now a random variable since it will be
computed from data realizations. Note that the value of
$\hat{\boldsymbol{f}}_{k,i}(l)$ may change over time due to
$\boldsymbol{b}_{k,i}(l)$.


Since all nodes have similar processing abilities, it is reasonable
to consider the following scenario.
\begin{assumption}[Homogeneous agents]\label{ass1}
All nodes in the network use the same step-size, $\mu_k=\mu$, and
they observe data arising from the same covariance distribution so
that $R_{u,k}=R_u$ for all $k$.
\end{assumption}
\noindent Agents still need to know whether to increase or decrease
the belief in (\ref{eq32}). We now suggest a procedure that allows
the nodes to estimate the vectors $\{f_k\}$ by focusing on their
behavior in the \emph{far-field} regime when their weight estimates
are usually far from their observed models (see (\ref{eq4}) for a
more specific description). The far-field regime generally occurs
during the initial stages of adaptation and, therefore, the vectors
$\{f_k\}$ can be determined quickly during these initial iterations.

To begin with, we refer to the update vector from (\ref{eq20}),
which can be written as follows for node $l$:
\begin{align}
    \boldsymbol{h}_{l,i} &=
    \mu^{-1}(\boldsymbol{\psi}_{l,i}-\boldsymbol{w}_{l,i-1})\notag
    \\ &=\boldsymbol{u}_{l,i}^T\boldsymbol{u}_{l,i}
    (z^\circ_l-\boldsymbol{w}_{l,i-1})+
    \boldsymbol{u}^T_{l,i}\boldsymbol{v}_l(i).\label{eq36}
\end{align}
Taking expectation of both sides conditioned on
$\boldsymbol{w}_{l,i-1}=w_{l,i-1}$, we have that
\begin{equation}\label{eq21}
    \bar{h}_{l,i}\triangleq
    \mathbb{E}[\boldsymbol{h}_{l,i}\;|\;\boldsymbol{w}_{l,i-1}=w_{l,i-1}]=
    R_{u}(z^\circ_l-w_{l,i-1}).
\end{equation}
That is, the expected update direction given the previous estimate,
$w_{l,i-1}$, is a scaled vector pointing from $w_{l,i-1}$ towards
$z^\circ_l$ with scaling matrix $R_{u}$. Note that since $R_{u}$ is
positive-definite, then the term $\bar{h}_{l,i}$ lies in the same
half plane of the vector $z^\circ_l-w_{l,i-1}$, i.e.,
$\bar{h}^T_{l,i}(z^\circ_l-w_{l,i-1})>0$. Therefore, the update
vector provides useful information about the observed model at node
$l$. For example, this term tells us how close the estimate at node
$l$ is to its observed model. When the magnitude of $\bar{h}_{l,i}$
is large, or the estimate at node $l$ is far from its observed model
$z^\circ_l$, then we say that node $l$ is in a \emph{far-field}
regime. On the other hand, when the magnitude of $\bar{h}_{l,i}$ is
small, then the estimate $w_{l,i-1}$ is close to $z_l^\circ$ and we
say that the node is operating in a \emph{near-field} regime. The
vector $\bar{h}_{l,i}$ can be estimated by the first-order
recursion:
\begin{align} \label{eq18}
    \hat{\boldsymbol{h}}_{l,i} = (1-\nu) \hat{\boldsymbol{h}}_{l,i-1} +
    \nu \mu^{-1}(\boldsymbol{\psi}_{l,i}-\boldsymbol{w}_{l,i-1})
\end{align}
where $\hat{\boldsymbol{h}}_{l,i}$ denotes the estimate for
$\bar{h}_{l,i}$ and $\nu$ is a positive step-size. Note that since
the value of $\bar{h}_{l,i}$ varies with $w_{l,i-1}$, which is
updated using the step-size $\mu$, then the value of $\nu$ should be
set large enough compared to $\mu$ (e.g., $\mu=0.005$ and $\nu=0.05$
are used in our simulations) so that recursion (\ref{eq18}) can
track variations in $\bar{h}_{l,i}$ over time. Moreover, since node
$k$ has access to the
$\{\boldsymbol{w}_{l,i-1},\boldsymbol{\psi}_{l,i}\}$ if node $l$ is
in its neighborhood, node $k$ can compute
$\hat{\boldsymbol{h}}_{l,i}$ on its own using (\ref{eq18}). In the
following, we describe how node $k$ updates the belief
$\boldsymbol{b}_{k,i}(l)$ using
$\{\hat{\boldsymbol{h}}_{k,i},\hat{\boldsymbol{h}}_{l,i}\}$.

\begin{figure}
\centering
\includegraphics[width=18em]{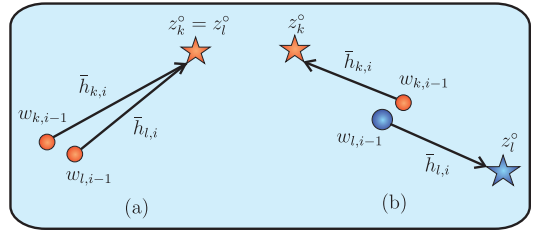}
\caption{Illustration of the vectors $\bar{h}_{k,i}$ and
$\bar{h}_{l,i}$ (a) when both nodes are in far-field and have the
same observed model or (b) different observed models.} \label{Fig_3}
\end{figure}

During the initial stage of adaptation, nodes $k$ and $l$ are
generally away from their respective observed models and both nodes
are therefore in the far-field. This state is characterized by the
conditions
\begin{equation}\label{eq4}
    \|\hat{\boldsymbol{h}}_{k,i}\| > \eta\quad \text{and}
    \quad \|\hat{\boldsymbol{h}}_{l,i}\|> \eta
\end{equation}
for some threshold $\eta$. If both nodes have the same observed
model, then the estimates $\hat{\boldsymbol{h}}_{k,i}$ and
$\hat{\boldsymbol{h}}_{l,i}$ are expected to have similar direction
towards the observed model (see Fig. \ref{Fig_3}(a)). Node $k$ will
increase the belief $\boldsymbol{b}_{k,i}(l)$ using (\ref{eq32}) if
\begin{align}\label{eq10}
    \hat{\boldsymbol{h}}_{k,i}^T\hat{\boldsymbol{h}}_{l,i}>0.
\end{align}
Otherwise, node $k$ will decrease the belief
$\boldsymbol{b}_{k,i}(l)$. That is, when both nodes are in the
far-field, then node $k$ increases its belief that node $l$ shares
the same observed model when the vectors
$\hat{\boldsymbol{h}}_{k,i}$ and $\hat{\boldsymbol{h}}_{l,i}$ lie in
the same quadrant. Note that it is possible for node $k$ to increase
$\boldsymbol{b}_{k,i}(l)$ even when nodes $k$ and $l$ have distinct
models. This is because it is difficult to differentiate between the
models during the initial stages of adaptation. This situation is
handled by the evolving network dynamics as follows. If node $k$
considers that the data from node $l$ originate from the same model,
then node $k$ will use the intermediate estimate
$\boldsymbol{\psi}_{l,i}$ from node $l$ in (\ref{eq15}). Eventually,
from Lemma \ref{lem3}, the estimates at these nodes get close to a
convex combination of the underlying models, which would then enable
node $k$ to distinguish between the two models and to decrease the
value of $\boldsymbol{b}_{k,i}(l)$. Clearly, for proper resolution,
the distance between the models needs to be large enough so that the
agents can resolve them. When the models are very close to each
other so that resolution is difficult, the estimates at the agents
will converge towards a convex combination of the models (which will
be also close to the models). Therefore, the belief
$\boldsymbol{b}_{k,i}(l)$ is updated according to the following
rule:
\begin{equation}\label{eq16}
    \boldsymbol{b}_{k,i}(l)=
    \begin{cases}
    \alpha \boldsymbol{b}_{k,i-1}(l) + (1-\alpha), &\text{if $E_1$}\\
    \alpha \boldsymbol{b}_{k,i-1}(l), &\text{if $E_1^c$}
    \end{cases}
\end{equation}
where $E_1$ and $E_1^c$ are the two events described by:
\begin{align}\label{eq116}
    E_1&: \|\hat{\boldsymbol{h}}_{k,i}\|>\eta,
    \|\hat{\boldsymbol{h}}_{l,i}\|>\eta,\text{ and }
    \hat{\boldsymbol{h}}_{k,i}^T\hat{\boldsymbol{h}}_{l,i}>0\\
    E_1^c&: \|\hat{\boldsymbol{h}}_{k,i}\|>\eta,
    \|\hat{\boldsymbol{h}}_{l,i}\|>\eta,\text{ and }
    \hat{\boldsymbol{h}}_{k,i}^T\hat{\boldsymbol{h}}_{l,i} \leq 0. \label{eq118}
\end{align}
Note that node $k$ updates the belief $b_{k,i}(l)$ only when both
nodes $k$ and $l$ are in the far-field.

\section{Diffusion Strategy with Decision-Making}\label{sec5}
Combining the modified diffusion strategy (\ref{eq6})-(\ref{eq15}),
the combination weights (\ref{eq24})-(\ref{eq25}), the
decision-making process (\ref{eq98})-(\ref{eq57}), and the
classification scheme (\ref{eq37}) and (\ref{eq16}) with
$\{f(l),g_i(l)\}$ replaced by
$\{\hat{\boldsymbol{f}}_{k,i}(l),\boldsymbol{g}_{k,i}(l)\}$, we
arrive at the listing shown in the table. It is seen from the
algorithm that the adaptation and combination steps of diffusion,
which correspond to steps 1) and 8), are now separated by several
steps. The purpose of these intermediate steps is to select the
combination weights properly to carry out the aggregation required
by step 8). Note that to implement the algorithm, nodes need to
exchange the quantities
$\{\boldsymbol{w}_{k,i-1},\boldsymbol{\psi}_{k,i},\boldsymbol{g}_{k,i-1}(k)\}$
with their neighbors. We summarize the computational complexity and
the amount of scalar exchanges of the conventional and modified
diffusion strategies in Table II. Note that the modified strategy
still requires in the order of $O(M)$ computations per iteration.
Nevertheless, the modified diffusion strategy requires about $2n_kM$
more additions and multiplications than conventional diffusion. This
is because of the need to compute the terms
$\{\hat{\boldsymbol{h}}_{l,i}\}$ in step 2). If the nodes can afford
to exchange extra information, then instead of every node connected
to node $l$ computing the term $\hat{\boldsymbol{h}}_{l,i}$ in step
2), this term can be computed locally by node $l$ and shared with
its neighbors. This reveals a useful trade-off between complexity
and information exchange.

\begin{algorithm}[t]{\footnotesize
\caption{(Diffusion strategy with decision-making)}
\begin{algorithmic}
\State For each node $k$, initialize $\boldsymbol{w}_{k,-1}=0$,
$\hat{\boldsymbol{h}}_{k,-1}=0$, $\boldsymbol{b}_{k,-1}(l)=0.5$, and
$\boldsymbol{g}_{k,-1}(k)=1$.

\For {$i\geq 0$ and $k=1$ to $N$}

\State 1) Perform an adaptation step using the local data
$\{\boldsymbol{d}_k(i), \boldsymbol{u}_{k,i}\}$:
\begin{align*}
    \boldsymbol{\psi}_{k,i} = \boldsymbol{w}_{k,i-1}+
    \mu\boldsymbol{u}_{k,i}^T
    [\boldsymbol{d}_k(i)-\boldsymbol{u}_{k,i}\boldsymbol{w}_{k,i-1}].
\end{align*}
\State 2) Exchange the vectors
$\{\boldsymbol{\psi}_{k,i},\boldsymbol{w}_{k,i-1}\}$ with
neighbors and update the \\
\hskip\algorithmicindent $\hspace{1.2em}$average update vectors
$\{\hat{\boldsymbol{h}}_{l,i}\}$ for $l\in\mathcal{N}_k$:
\begin{align*}
    \hat{\boldsymbol{h}}_{l,i}=(1-\nu)\hat{\boldsymbol{h}}_{l,i-1}+
    \nu \mu^{-1}(\boldsymbol{\psi}_{l,i}-\boldsymbol{w}_{l,i-1}).
\end{align*}
\State 3) Update the beliefs $\{\boldsymbol{b}_{k,i}(l)\}$ for
$l\in\mathcal{N}_k\setminus\{k\}$:
\begin{equation*}
    \boldsymbol{b}_{k,i}(l)=
    \begin{cases}
    \alpha \boldsymbol{b}_{k,i-1}(l) + (1-\alpha), &\text{if $E_1$}\\
    \alpha \boldsymbol{b}_{k,i-1}(l), &\text{if $E_1^c$}
    \end{cases}
\end{equation*}
\hskip\algorithmicindent $\hspace{1.2em}$where $E_1$ and $E_1^c$ are
defined in (\ref{eq116})-(\ref{eq118}).

\State 4) Identify the observed models
$\{\hat{\boldsymbol{f}}_{k,i}(l)\}$ for
$l\in\mathcal{N}_k\setminus\{k\}$:
\begin{equation*}
    \hat{\boldsymbol{f}}_{k,i}(l)=
    \begin{cases}
    1, &\text{if $\boldsymbol{b}_{k,i}(l) \geq 0.5$}\\
    0, &\text{otherwise}
    \end{cases}.
\end{equation*}
\State 5) Collect the desired models $\{\boldsymbol{g}_{k,i-1}(l)\}$
for
$l\in\mathcal{N}_k\setminus\{k\}$ and \\
\hskip\algorithmicindent $\hspace{1.2em}$construct the set
$\mathcal{N}^g_{k,i-1}$ as follows:
\begin{align*}
    \boldsymbol{g}_{k,i-1}(l) &=
    \begin{cases}
    \boldsymbol{g}_{l,i-1}(l), &\text{if $\hat{\boldsymbol{f}}_{k,i}(l) = 1$}\\
    1-\boldsymbol{g}_{l,i-1}(l), &\text{otherwise}
    \end{cases}\\
    \mathcal{N}_{k,i-1}^g &=
    \{l\;|\;l\in\mathcal{N}_k,\boldsymbol{g}_{k,i-1}(l)=\boldsymbol{g}_{k,i-1}(k)\}.
\end{align*}
\State 6) Update the desired model $\boldsymbol{g}_{k,i}(k)$:
\begin{equation*}\label{eq19}
    \boldsymbol{g}_{k,i}(k) = \begin{cases}
    \boldsymbol{g}_{k,i-1}(k), &\text{w.p. $q_{k,i-1}$}\\
    1-\boldsymbol{g}_{k,i-1}(k), &\text{w.p. $1-q_{k,i-1}$}
    \end{cases}
\end{equation*}
\hskip\algorithmicindent $\hspace{1.2em}$where the probability
$q_{k,i-1}$ is defined in (\ref{eq57}).

\State 7) Adjust the combination weights $\{a^{(1)}_{l,k}\}$ and
$\{a^{(2)}_{l,k}\}$:
\begin{align*}
    & a^{(1)}_{l,k,i} =
    \begin{cases}
    a_{l,k}, &\text{if $l\in\mathcal{N}_k$ and
    $\hat{\boldsymbol{f}}_{k,i}(l)= \boldsymbol{g}_{k,i}(k)$}\\
    0, &\text{otherwise}
    \end{cases}\\
    & a^{(2)}_{l,k,i} =
    \begin{cases}
    a_{l,k}, &\text{if $l\in\mathcal{N}_k$ and
    $\hat{\boldsymbol{f}}_{k,i}(l)\neq \boldsymbol{g}_{k,i}(k)$}\\
    0, &\text{otherwise}
    \end{cases}.
\end{align*}
\State 8) Perform the combination step:
\begin{align*}
    \boldsymbol{w}_{k,i}= \sum_{l\in\mathcal{N}_k}\left(a^{(1)}_{l,k,i}
    \boldsymbol{\psi}_{l,i}
    +a^{(2)}_{l,k,i}\boldsymbol{w}_{l,i-1}\right).
\end{align*}
\EndFor
\end{algorithmic}}
\end{algorithm}

\begin{table}
\centering \caption{{\rm Comparison of the number of multiplications
and additions per iteration, as well as the number of scalars that
are exchanged for each iteration of the algorithms at every node
$k$. In the table, the symbol $n_k$ denotes the degree of node $k$,
i.e., the size of its neighborhood $\mathcal{N}_k$.}}
\begin{tabular}{|c|c|c|}
\hline & \textbf{Diffusion} (\ref{eq49})-(\ref{eq1}) &
\textbf{Modified
diffusion}  \\
\hline\hline \textbf{Multiplications} &
    $(n_k+2)M$ & $(3n_k+2)M+n_k-1$ \\
\hline \textbf{Additions} &
    $(n_k+1)M$ & $(3n_k+1)M+n_k-1$ \\
\hline \textbf{Scalar exchanges} &
    $n_kM$ & $n_k(2M+1)$ \\
\hline
\end{tabular}
\end{table}

Due to the dependency among the steps of the algorithm, the analysis
of its behavior becomes challenging. However, by examining the
various steps, some useful observations stand out. Specifically, it
is observed that the convergence of the algorithm occurs in three
phases as follows (see also Sec. VIII):
\begin{enumerate}
\item
Convergence of the classification scheme: The first phase of
convergence happens during the initial stages of adaptation. It is
natural to expect that during this stage, all weight estimates are
generally away from their respective models and the nodes operate in
the far-field regime. Then, the nodes use steps 2)-5) to determine
the observed models $\{\hat{\boldsymbol{f}}_{k,i}(l)\}$ of their
neighbors. We explain later in Eq. (\ref{eq99}) in Theorem
\ref{thm5} that this construction is able to identify the observed
models with high probability. In other words, the classification
scheme is able to converge reasonably well and fast during the
initial stages of adaptation.

\item
Convergence of the decision-making process: The second phase of
convergence happens right after the convergence of the
classification scheme, once the $\{\hat{\boldsymbol{f}}_{k,i}(l)\}$
have converged. Because the nodes now have correct information about
their neighbor's observed models, they use steps 5)-6) to determine
their own desired models $\{\boldsymbol{g}_{k,i}(k)\}$. The
convergence of this step is ensured by Eq. (\ref{eq127}) in Theorem
\ref{thm1}.

\item Convergence of the diffusion strategy: After the classification
and decision-making processes converge, the estimates
$\{\hat{\boldsymbol{f}}_{k,i}(l),\boldsymbol{g}_{k,i}(l)\}$ remain
largely invariant and the combination weights in step 7) therefore
remain fixed for all practical purposes. Then, the diffusion
strategy becomes unbiased and converges in the mean according to
Theorem \ref{thm2}. Moreover, when the estimates are close to
steady-state, those nodes whose observed models are the same as the
desired model enter the near-field regime and they stop updating
their belief vectors (this will be justified by the future result
(\ref{eq74})).
\end{enumerate}

\section{Performance of Classification Procedure}\label{sec3}
It is clear that the success of the diffusion strategy and
decision-making process depends on the reliability of the
classification scheme in (\ref{eq37}) and (\ref{eq16}). In this
section, we examine the probability of error for the classification
scheme under some simplifying conditions to facilitate the analysis.
This is a challenging task to pursue due to the stochastic nature of
the classification and decision-making process, and due to the
coupling among the agents. Our purpose in this section is to gain
some insights into this process through a first-order approximate
analysis.

Now, there are two types of error. When nodes $k$ and $l$ are
subject to the same observed model (i.e., $z^\circ_k=z^\circ_l$ and
$f_{k}(l)=1$), then one probability of error is defined as:
\begin{align}
    P_{e,1}&=\Pr\left(\hat{\boldsymbol{f}}_{k,i}(l)=0
    \;|\;f_{k}(l)=1\right)\notag \\
    &=\Pr\left(\boldsymbol{b}_{k,i}(l)<0.5
    \;|\;z^\circ_k=z^\circ_l\right)\label{eq77}
\end{align}
where we used rule (\ref{eq37}). The second type of probability of
error occurs when both nodes have different observed models (i.e.,
when $z^\circ_k\neq z^\circ_l$ and $f_{k}(l)=0$) and refers to the
case:
\begin{align}
    P_{e,0}&=\Pr\left(\hat{\boldsymbol{f}}_{k,i}(l) = 1
    \;|\;f_{k}(l)=0\right)\notag \\
    &=\Pr\left(\boldsymbol{b}_{k,i}(l) >
    0.5\;|\;z^\circ_k\neq z^\circ_l\right).\label{eq92}
\end{align}
To evaluate the error probabilities in (\ref{eq77})-(\ref{eq92}), we
examine the probability distribution of the belief variable
$\boldsymbol{b}_{k,i}$. Note from (\ref{eq16}) that the belief
variable can be expressed as:
\begin{equation}\label{eq46}
    \boldsymbol{b}_{k,i}(l) = \alpha \boldsymbol{b}_{k,i-1}(l)
    +(1-\alpha)\boldsymbol{\xi}_{k,i}(l)
\end{equation}
where $\boldsymbol{\xi}_{k,i}(l)$ is a Bernoulli random variable
with
\begin{equation} \label{eq107}
    \boldsymbol{\xi}_{k,i}(l) =
    \begin{cases}
    1, &\text{with probability $p$}\\
    0, &\text{with probability $1-p$}
    \end{cases}.
\end{equation}
The value of $p$ depends on whether the nodes have the same observed
models or not. When $z^\circ_k=z^\circ_l$, the belief
$\boldsymbol{b}_{k,i}(l)$ is supposed to be increased and the
probability of detection, $P_d$, characterizes the probability that
$\boldsymbol{b}_{k,i}(l)$ is increased, i.e.,
\begin{equation}\label{eq100}
    P_d = \Pr(\boldsymbol{\xi}_{k,i}(l)=1
    \;|\; z^\circ_k=z^\circ_l).
\end{equation}
In this case, the probability $p$ in (\ref{eq107}) will be replaced
by $P_d$. On the other hand, when $z^\circ_k\neq z^\circ_l$, the
probability of false alarm, $P_f$, characterizes the probability
that the belief $\boldsymbol{b}_{k,i}(l)$ is increased when it is
supposed to be decreased, i.e.,
\begin{equation}\label{eq101}
    P_f = \Pr(\boldsymbol{\xi}_{k,i}(l)=1
    \;|\; z^\circ_k\neq z^\circ_l)
\end{equation}
and we replace $p$ in (\ref{eq107}) by $P_f$. We will show later
(see Lemma \ref{lem6}) how to evaluate the two probabilities $P_d$
and $P_f$. In the sequel we denote them generically by $p$.

Expanding (\ref{eq46}), we obtain
\begin{align}\label{eq71}
    \boldsymbol{b}_{k,i}(l)=\alpha^{i+1}\boldsymbol{b}_{k,-1}(l)+
    (1-\alpha)\sum_{j=0}^{i}\alpha^j\boldsymbol{\xi}_{k,i-j}(l).
\end{align}
Although it is generally not true, we simplify the analysis by
assuming that the $\{\boldsymbol{\xi}_{k,i}(l)\}$ in (\ref{eq107})
are independent and identically distributed (i.i.d.) random
variables. This assumption is motivated by conditions
(\ref{eq116})-(\ref{eq118}) and by the fact that the type of model
that is observed by node $k$ is assumed to be independent from the
type of model that is observed by node $l$. The assumption is also
motivated by the fact that the regression data and noise across all
nodes are assumed to be temporally white and independent over space.
Now, since $\alpha$ is a design parameter that is smaller than one,
after a few iterations, say, $C$ iterations, the influence of the
initial condition in (\ref{eq71}) becomes small and can be ignored.
In addition, the distribution of $\boldsymbol{b}_{k,i}(l)$ can be
approximated by the distribution of the following random variable,
which takes the form of a random geometric series:
\begin{equation}\label{eq117}
    \boldsymbol{\zeta}_k(l) \triangleq(1-\alpha)
    \sum_{j=0}^C\alpha^j\boldsymbol{\xi}_{k,j}(l)
\end{equation}
where we replaced the index $i-j$ in (\ref{eq71}) by $j$ because the
$\{\boldsymbol{\xi}_{k,i}(l)\}$ are assumed to be i.i.d. There have
been several useful works on the distribution function of random
geometric sequences and series \cite{Hill73,Smith91,Bovier93}.
However, it is generally untractable to express the distribution
function in closed form. We instead resort to the following two
inequalities to establish bounds for the error probabilities
(\ref{eq77})-(\ref{eq92}). First, for any two generic events $E_1$
and $E_2$, if $E_1$ implies $E_2$, then the probability of event
$E_1$ is less than the probability of event $E_2$ \cite{Papoulis02},
i.e.,
\begin{equation}\label{eq60}
    \Pr(E_1)\leq \Pr(E_2)\quad\text{if}\quad E_1\subseteq E_2.
\end{equation}
The second inequality is the Markov inequality \cite{Papoulis02},
i.e., for any nonnegative random variable $\boldsymbol{x}$ and
positive scalar $\delta$, it holds that
\begin{equation}\label{eq61}
    \Pr(\boldsymbol{x} \geq \delta)=\Pr(\boldsymbol{x}^2 \geq \delta^2)
    \leq \frac{\mathbb{E}\boldsymbol{x}^2}{\delta^2}.
\end{equation}
To apply the Markov inequality (\ref{eq61}), we need the
second-order moment of $\boldsymbol{\zeta}_k(l)$ in (\ref{eq117}),
which is difficult to evaluate because the
$\{\boldsymbol{\xi}_{k,j}(l)\}$ are not zero mean. To circumvent
this difficulty, we introduce the change of variable:
\begin{equation}
    \boldsymbol{\xi}^\circ_{k,j}(l) \triangleq
    \frac{\boldsymbol{\xi}_{k,j}(l)-p}{\sqrt{p(1-p)}}.
\end{equation}
It can be verified that the $\{\boldsymbol{\xi}^\circ_{k,j}(l)\}$
are i.i.d. with \emph{zero} mean and unit variance. Then, we can
write (\ref{eq117}) as
\begin{align}\label{eq93}
    \boldsymbol{\zeta}_k(l)
    = p\left(1-\alpha^{C+1}\right) + \sqrt{p(1-p)}\boldsymbol{\zeta}^\circ_k(l)
\end{align}
where the variable $\boldsymbol{\zeta}^\circ_k(l)$ is defined by
\begin{equation}
    \boldsymbol{\zeta}^\circ_k(l) \triangleq (1-\alpha)\sum_{j=0}^C\alpha^j
    \boldsymbol{\xi}^\circ_{k,j}(l)
\end{equation}
and its mean is zero and its variance is given by
\begin{align}
    \mathbb{E}(\boldsymbol{\zeta}^\circ_k(l))^2
    =\frac{1-\alpha}{1+\alpha}\left(1-\alpha^{2(C+1)}\right)
    \approx\frac{1-\alpha}{1+\alpha}.
\end{align}
Then, from (\ref{eq77}) and (\ref{eq93}) and replacing the
probability $p$ by $P_d$ and for $C$ large enough so that
$1-\alpha^{C+1}\approx 1$, we obtain that
\begin{align}
    P_{e,1}
    & \approx \Pr(\boldsymbol{\zeta}_k(l) < 0.5\;|\; z^\circ_k=z^\circ_l) \notag\\
    & = \Pr\left(\boldsymbol{\zeta}^\circ_k(l) < \frac{-(P_d-0.5)}
    {\sqrt{P_d(1-P_d)}}\;|\;z^\circ_k=z^\circ_l\right) \notag \\
    & \leq \Pr\left(|\boldsymbol{\zeta}^\circ_k(l)| > \frac{P_d-0.5}
    {\sqrt{P_d(1-P_d)}}\;|\;z^\circ_k=z^\circ_l\right) \notag \\
    & \leq \frac{1-\alpha}{1+\alpha}\cdot
    \frac{P_d(1-P_d)}{(P_d-0.5)^2}\label{eq78}
\end{align}
where we used (\ref{eq60}) and the Markov inequality (\ref{eq61}) in
the last two inequalities. Note that in (\ref{eq78}), we assume the
value of $P_d$ to be greater than $0.5$. Indeed, as we will argue in
Lemma \ref{lem6}, the value of $P_d$ is close to one. Similarly,
replacing the probability $p$ by $P_f$ and assuming that $P_f<0.5$,
we obtain from (\ref{eq92}) and (\ref{eq93}) that
\begin{align}
    P_{e,0}\leq \frac{1-\alpha}{1+\alpha}\cdot
    \frac{P_f(1-P_f)}{(0.5-P_f)^2}.\label{eq79}
\end{align}
To evaluate the upper bounds in (\ref{eq78})-(\ref{eq79}), we need
the probabilities of detection and false alarm in
(\ref{eq100})-(\ref{eq101}). Since the update of
$\boldsymbol{b}_{k,i}(l)$ in (\ref{eq16}) depends on
$\{\hat{\boldsymbol{h}}_{k,i},\hat{\boldsymbol{h}}_{l,i}\}$, we need
to rely on the statistical properties of these latter quantities. In
the following, we first examine the statistics of
$\hat{\boldsymbol{h}}_{k,i}$ constructed via (\ref{eq18}) and then
evaluate $P_d$ and $P_f$ defined by (\ref{eq100})-(\ref{eq101}).

\subsection{Statistics of $\hat{\boldsymbol{h}}_{k,i}$}
We first discuss some assumptions that lead to an approximate model
for the evolution of $\hat{\boldsymbol{h}}_{k,i}$ in (\ref{eq75})
further ahead. As we mentioned following (\ref{eq18}), since the
step-sizes $\{\mu,\nu\}$ satisfy $\mu\ll \nu$, the variation of
$\boldsymbol{w}_{k,i-1}$ can be assumed to be much slower than the
variation of $\hat{\boldsymbol{h}}_{k,i}$. For this reason, the
analysis in this section will be conditioned on
$\boldsymbol{w}_{k,i-1}=w_{k,i-1}$, as we did in (\ref{eq21}), and
we introduce the following assumption.
\begin{assumption}[Small step-size]\label{ass2}
The step-sizes $\{\mu,\nu\}$ are sufficiently small, i.e.,
\begin{equation} \label{eq70}
    0 < \mu \ll \nu \ll 1
\end{equation}
so that $w_{k,i} \approx w_{k,i-1}$ for all $k$.
\end{assumption}
\noindent In addition, since the update vector from (\ref{eq21})
depends on the covariance matrix $R_u$, we assume $R_u$ is
well-conditioned so that the following is justified.
\begin{assumption}[Regression model]\label{ass5}
The regression covariance matrix $R_{u}$ is well-conditioned such
that it holds that
\begin{align}\label{eq72}
    &\text{if }\|z^\circ_k-w_{k,i-1}\| \gg 1\text{, then }
    \|\bar{h}_{k,i}\| \gg \eta\\
    &\text{if }\|z^\circ_k-w_{k,i-1}\| \ll 1\text{, then }
    \|\bar{h}_{k,i}\| \ll \eta.\label{eq73}
\end{align}
Moreover, the fourth-order moment of the regression data
$\{\boldsymbol{u}_{k,i}\}$ is assumed to be bounded such that
\begin{equation}\label{eq108}
    \nu \tau \ll 1
\end{equation}
where the scalar $\tau$ is a bound for
\begin{equation}\label{eq102}
    \frac{
    \mathbb{E}\|\boldsymbol{u}_{k,i}^T\boldsymbol{u}_{k,i}
    (z_k^\circ-w_{k,i-1})-\bar{h}_{k,i}\|^2}{\|\bar{h}_{k,i}\|^2}
    \leq \tau
\end{equation}
and its value measures the randomness in variables involving
fourth-order products of entries of $\boldsymbol{u}_{k,i}$.
\end{assumption}
\noindent Note that condition (\ref{eq102}) can be rewritten as
\begin{equation}
    \frac{(z_k^\circ-w_{k,i-1})^T\mathbb{E}(\boldsymbol{u}_{k,i}^T\boldsymbol{u}_{k,i}
    \boldsymbol{u}_{k,i}^T\boldsymbol{u}_{k,i}-R_u^2)(z_k^\circ-w_{k,i-1})}
    {(z_k^\circ-w_{k,i-1})^TR_u^2(z_k^\circ-w_{k,i-1})}\leq \tau
\end{equation}
which shows that (\ref{eq102}) corresponds to a condition on the
fourth-order moment of the regression data. Combining conditions
(\ref{eq70}) and (\ref{eq108}), we obtain the following constraint
on the step-sizes $\{\mu,\nu\}$:
\begin{equation}\label{eq126}
    0 \ll \mu \ll \nu \ll \min\{1,1/\tau\}.
\end{equation}
To explain more clearly what conditions (\ref{eq72})-(\ref{eq73})
entail, we obtain from (\ref{eq21}) that $\|\bar{h}_{k,i}\|^2$ can
be written as the weighted square Euclidean norm:
\begin{equation}\label{eq76}
    \|\bar{h}_{k,i}\|^2 = \|z^\circ_k-w_{k,i-1}\|^2_{R_u^2}.
\end{equation}
We apply the Rayleigh-Ritz characterization of eigenvalues
\cite{Horn85} to conclude that
\begin{equation*}
    \lambda_{\min}(R_u) \cdot\|z^\circ_k-w_{k,i-1}\|
    \leq \|\bar{h}_{k,i}\| \leq \lambda_{\max}(R_u) \cdot\|z^\circ_k-w_{k,i-1}\|
\end{equation*}
where $\lambda_{\min}(R_u)$ and $\lambda_{\max}(R_u)$ denote the
minimum and maximum eigenvalues of $R_u$. Then, condition
(\ref{eq72}) indicates that whenever node $k$ is operating in the
far-field regime, i.e., whenever $\|z^\circ_k-w_{k,i-1}\|\gg 1$,
then we would like
\begin{equation}
    \lambda_{\min}(R_u) \cdot\|z^\circ_k-w_{k,i-1}\| \gg \eta.
\end{equation}
Likewise, whenever $\|z_k^\circ-w_{k,i-1}\|\ll 1$, then
\begin{equation}
    \lambda_{\max}(R_u) \cdot\|z^\circ_k-w_{k,i-1}\| \ll \eta.
\end{equation}
Therefore, the scalars $\lambda_{\min}(R_u)/\eta$ and
$\lambda_{\max}(R_u)/\eta$ cannot be too small or too large, i.e.,
the matrix $R_u$ should be well-conditioned.

We are now ready to model the average update vector
$\hat{\boldsymbol{h}}_{k,i}$. From Assumption \ref{ass2}, since the
estimate $w_{k,i-1}$ remains approximately constant during repeated
updates of $\hat{\boldsymbol{h}}_{k,i}$, we first remove the time
index in $w_{k,i-1}$ and examine the statistics of
$\hat{\boldsymbol{h}}_{k,i}$ under the condition $w_{k,i-1} = w_k$.
From (\ref{eq36}) and (\ref{eq18}), the expected value of
$\hat{\boldsymbol{h}}_{k,i}$ given $w_{k,i-1} = w_k$ converges to
\begin{equation}\label{eq63}
    \lim_{i\rightarrow \infty}\mathbb{E}\hat{\boldsymbol{h}}_{k,i} =
    R_{u}(z_k^\circ-w_k)
    \triangleq \bar{h}_k.
\end{equation}
We can also obtain from (\ref{eq36}) and (\ref{eq18}) that the
limiting second-order moment of $\hat{\boldsymbol{h}}_{k,i}$, which
is denoted by $\sigma^2_{\hat{h},k}$, satisfies:
\begin{align} \label{eq39}
    \sigma^2_{\hat{h},k}\triangleq \lim_{i\rightarrow \infty}
    \mathbb{E}\|\hat{\boldsymbol{h}}_{k,i}-\bar{h}_k\|^2
    = (1-\nu)^2 \sigma^2_{\hat{h},k}+ \nu^2 \sigma^2_{h,k}
\end{align}
where $\sigma^2_{h,k}\triangleq
\mathbb{E}\|\boldsymbol{h}_{k,i}-\bar{h}_k\|^2$ is given by
\begin{align}
    \sigma^2_{h,k}=\mathbb{E}\|\boldsymbol{u}_{k,i}^T\boldsymbol{u}_{k,i}
    (z_k^\circ-w_k)-\bar{h}_k\|^2+
    \sigma^2_{v,k}\text{Tr}(R_{u}). \label{eq33}
\end{align}
Note that the cross term on the right-hand side of (\ref{eq39}) is
zero because the terms $\hat{\boldsymbol{h}}_{k,i-1}-\bar{h}_k$ and
$\boldsymbol{h}_{k,i}-\bar{h}_k$ are independent under the constant
$w_k$ condition. Note also that $\boldsymbol{h}_{k,i}-\bar{h}_k$ has
zero mean. Then, from (\ref{eq39}) and Assumption \ref{ass2}, the
variance $\sigma^2_{\hat{h},k}$ is given by
\begin{align}\label{eq59}
    \sigma^2_{\hat{h},k}=\frac{\nu}{2-\nu}\sigma^2_{h,k}
    \approx\frac{\nu}{2}\sigma^2_{h,k}.
\end{align}

Since $w_{k,i-1}$ remains approximately constant, the average update
vector $\hat{\boldsymbol{h}}_{k,i}$ has mean and second-order moment
close to expressions (\ref{eq63}) and (\ref{eq59}). We then arrive
at the following approximate model for $\hat{\boldsymbol{h}}_{k,i}$.
\begin{assumption}[Model for
$\hat{\boldsymbol{h}}_{k,i}$]\label{ass6} The estimate
$\hat{\boldsymbol{h}}_{k,i}$ is modeled as:
\begin{equation}\label{eq75}
    \hat{\boldsymbol{h}}_{k,i} = \bar{h}_{k,i} + \boldsymbol{n}_{k,i}
\end{equation}
where $\boldsymbol{n}_{k,i}$ is a random perturbation process with
zero mean and
\begin{equation}\label{eq40}
    \mathbb{E}\|\boldsymbol{n}_{k,i}\|^2 \leq
    \frac{\nu[\tau\|\bar{h}_{k,i}\|^2+\sigma^2_{v,k}\mathrm{Tr}(R_{u})]}{2}
\end{equation}
with the scalar $\tau$ defined by (\ref{eq108}).
\end{assumption}
\noindent Note that since the perturbation $\boldsymbol{n}_{k,i}$ is
from the randomness of the regressor and noise processes
$\{\boldsymbol{u}_{k,i},\boldsymbol{v}_{k}(i)\}$, then it is
reasonable to assume that the $\{\boldsymbol{n}_{k,i}\}$ are
independent of each other.

Before we proceed to the probability of detection (\ref{eq100}) and
the probability of false alarm (\ref{eq101}), we note that the
update of the belief $\boldsymbol{b}_{k,i}(l)$ happens only when
both nodes $k$ and $l$ are in the far-field regime, which is
determined by the magnitudes of $\hat{\boldsymbol{h}}_{k,i}$ and
$\hat{\boldsymbol{h}}_{l,i}$ being greater than the threshold
$\eta$. The following result approximates the probability that a
node is classified to be in the far-field.
\begin{lem}\label{lem5}
Under Assumptions \ref{ass2}-\ref{ass6}, it holds that
\begin{align}\label{eq62}
    \Pr(\|\hat{\boldsymbol{h}}_{k,i}\| > \eta\;|\;
    \|z^\circ_k-w_{k,i-1}\| \gg 1) &\geq 1-\frac{\nu \tau}{2} \\
    \Pr(\|\hat{\boldsymbol{h}}_{k,i}\| > \eta\;|\;
    \|z^\circ_k-w_{k,i-1}\| \ll 1) &\leq
    \frac{\nu \sigma^2_{v,k}\mathrm{Tr}(R_u)}{2\eta^2}. \label{eq74}
\end{align}
\end{lem}
\begin{proof}
See Appendix \ref{appA}.
\end{proof}
From Assumptions \ref{ass2}-\ref{ass5}, the probability in
(\ref{eq62}) is close to one and the probability in (\ref{eq74}) is
close to zero. Therefore, this approximate analysis suggests that
during the initial stages of adaptation, the magnitude of
$\{\|\hat{\boldsymbol{h}}_{k,i}\|\}$ successfully determines that
the nodes are in the far-field state and they update the belief
using rule (\ref{eq16}). When the estimates approach steady-state,
the nodes whose observed models are the same as the desired model
satisfy the condition $\|z^\circ_k-w_{k,i-1}\| \ll 1$ and,
therefore, they stop updating their belief vectors in view of
(\ref{eq74}). On the other hand, when both nodes $k$ and $l$ have
observed models that are different from the desired model (and,
therefore, their estimates are away from their observed models),
they will continue to update their beliefs. The proof in Appendix
\ref{appB} then establishes the following bounds on $P_d$ and $P_f$.
\begin{lem}\label{lem6}
Under Assumptions \ref{ass2}-\ref{ass6} and during the far-field
regime (\ref{eq72}), the probabilities of detection and false alarm
defined by (\ref{eq100})-(\ref{eq101}) are approximately bounded by
\begin{align}\label{eq68}
    P_d \geq 1-\frac{16\nu \tau}{\pi^2} \quad \text{and}\quad
    P_f \leq \frac{16\nu \tau}{\pi^2}.
\end{align}
\end{lem}
\noindent The above result establishes that the probability of
detection is close to one and the probability of false alarm is
close to zero in view of $\nu \tau \ll 1$. That is, with high
probability, node $k$ will correctly adjust the value of
$\boldsymbol{b}_{k,i}(l)$. We then arrive at the following bound for
error probabilities in (\ref{eq77})-(\ref{eq92}).
\begin{thm}\label{thm5}
Under Assumptions \ref{ass2}-\ref{ass6} and in the far-field regime
(\ref{eq72}), the error probabilities $\{P_{e,1},P_{e,0}\}$ are
approximately upper bounded by
\begin{equation}\label{eq99}
    P_u=\frac{1-\alpha}{1+\alpha}\cdot\frac{16\nu \tau}{\pi^2}
    \cdot \frac{1-16\nu \tau/\pi^2}{(1/2-16\nu \tau/\pi^2)^2}
    =\mathcal{O}(\nu).
\end{equation}
\end{thm}
\begin{proof}
Let the function $f(p)$ be defined as $p(1-p)/(p-0.5)^2$. It can be
verified that the function $f(p)$ is strictly increasing when
$p\in[0,0.5)$ and strictly decreasing when $p\in(0.5,1]$. From Lemma
\ref{lem6}, we conclude that $P_d>0.5$ and $P_f<0.5$. Therefore, an
upper bound for $P_{e,1}$ can be obtained by replacing $P_d$ in
(\ref{eq78}) by the lower bound in (\ref{eq68}). Similar arguments
apply to the upper bound for $P_{e,0}$.
\end{proof}
\noindent This result reveals that the $\{P_{e,1},P_{e,0}\}$ are
upper bounded by the order of $\nu$. In addition, the upper bound
$P_u$ also depends on the value of $\alpha$ used to update the
belief in (\ref{eq16}). We observe that the larger the value of
$\alpha$, the smaller the values of the error probabilities. In
simulations, we choose $\nu=0.05$ and $\alpha=0.95$, which will give
the upper bound in (\ref{eq99}) the value $P_u\approx 0.008\tau <
\nu \tau$. This implies that the classification scheme (\ref{eq37})
identifies the observed models with high probability.

\section{Rates of Convergence}
There are two rates of convergence to consider for adaptive networks
running a decision-making process of the form described in the
earlier sections. First, we need to analyze the rate at which the
nodes reach an agreement on a desired model (which corresponds to
the speed of the decision-making process). Second, we analyze the
rate at which the estimates by the nodes converge to the desired
model (which corresponds to the speed of the diffusion adaptation).

\subsection{Convergence Rate of Decision-Making Process}
From the proof of  Theorem \ref{thm1} (see Appendix \ref{appD}), the
decision-making process can be modeled as a Markov chain with $N+1$
states $\{\chi_i\}$ corresponding to the number of nodes whose
desired vectors are $w^\circ_1$. The Markov chain has two absorbing
states $\{0,N\}$ and its transition probability matrix $P$ can be
written as:
\begin{equation}\label{eq81}
    P = \begin{bmatrix}
    1 & 0 & 0 \\
    b & Q & c \\
    0 & 0 & 1
    \end{bmatrix}
\end{equation}
where the matrix $Q$ of size $(N-1)\times(N-1)$ is the transition
matrix among the transient states $\{1,2,\cdots,N-1\}$, and the
vectors $\{b,c\}$ of size $N-1$ are the transition probabilities
from the transient states to the absorbing states. The convergence
rate of the decision-making process is then determined by the rate
at which, starting at any arbitrary transient state, the Markov
chain converges to one of the absorbing states. The argument that
follows is meant to show that the rate of convergence of the
decision making process improves with the parameter $K$ used in
(\ref{eq57}); the larger the value of $K$ the faster is the
convergence.

To arrive at this conclusion, we first remark that to assess  the
rate of convergence, we need to compute the $j$th power of $P$ from
(\ref{eq81}) to find that
\begin{equation}
    P^j = \begin{bmatrix}
    1 & 0 & 0 \\
    \bar{b} & Q^j & \bar{c} \\
    0 & 0 & 1
    \end{bmatrix}
\end{equation}
where $\{\bar{b},\bar{c}\}$ are two $N\times 1$ vectors. Let the
Markov chain start from any arbitrary initial state distribution,
$y$, of the form
\begin{equation}
    y^T = \begin{bmatrix}
    0 & y^T_Q & 0
    \end{bmatrix}
\end{equation}
where $y_Q$ is a vector of size $N-1$ and its entries add up to one,
i.e., $y_Q^T\mathds{1}_{N-1} = 1$. We shall select $y_Q$ in a manner
that enables us to determine how the convergence rate depends on
$K$. Thus, note that the state distribution after $j$ transitions
becomes
\begin{equation}
    y^TP^j =
    \begin{bmatrix}
    y^T_Q\bar{b} & y^T_QQ^j & y^T_Q\bar{c}
    \end{bmatrix}.
\end{equation}
Therefore, the convergence rate is measured by the rate at which the
matrix $Q^j$ converges to zero, which is determined by the spectral
radius of $Q$. Since $Q$ is the sub-matrix of the transition
probability matrix, all entries of $Q$ are nonnegative, then by the
Perron-Frobenius Theorem \cite{Horn85}, the vector $y_Q$ can be
selected to be the left eigenvector of $Q$ corresponding to the
eigenvalue $\rho(Q)$, i.e., $y_Q^T Q = \rho(Q)y_Q^T$. Moreover, from
(\ref{eq123}), the matrix $Q$ is primitive and, therefore, all
entries of $y_Q$ are positive. Furthermore, since the transition
probability matrix $P$ is right-stochastic (i.e.,
$P\mathds{1}_{N+1}=\mathds{1}_{N+1}$), from (\ref{eq81}) it holds
that
\begin{equation} \label{eq82}
    b + c + Q\mathds{1}_{N-1} = \mathds{1}_{N-1}.
\end{equation}
Pre-multiplying the vector $y_Q$ on both sides of (\ref{eq82}), we
obtain that the convergence rate of the decision-making process can
be determined by
\begin{align}
    \rho(Q) = y_Q^TQ\mathds{1}_{N-1}
    = 1-y_Q^T(b+c).\label{eq52}
\end{align}
We now determine the value of the vector sum $b+c$. We note from
(\ref{eq121}) that the transition probabilities $\{p_{n,m}\}$ in $Q$
are determined by the probability $q_{k,i-1}$ from (\ref{eq57}), so
is the spectral radius of $Q$. We further note from (\ref{eq57})
that there is a single parameter $K$ dictating the value of
$q_{k,i-1}$. In the following, we examine the dependence of the
convergence rate $\rho(Q)$ on the parameter $K$. It is generally
challenging to develop the relation because the transition
probability $p_{n,m}$ needs to be computed in a compounded way where
we need to evaluate the summation of the products of
$\{q_{k,i-1}\}$. Nevertheless, some useful insights can be obtained
by means of the following approximate argument. Suppose the network
size is sufficiently large and that the nodes are uniformly
distributed in the spatial domain so that each of the nodes in the
network has approximately the same number of neighbors collecting
data from model $w_{1}^\circ$; likewise, each of the nodes in the
network has approximately the same number of neighbors collecting
data from the other model. Suppose that there are $\chi_{i-1}=n$ out
of $N$ nodes with desired model $w^\circ_1$, then, on average, node
$k$ with $n_k$ neighbors will have $n_kn/N$ neighbors whose desired
model is $w^\circ_1$ and have $n_k(1-n/N)$ neighbors whose desired
model is $w^\circ_0$. Then, from rule (\ref{eq17})-(\ref{eq57}),
node $k$ chooses $w^\circ_1$ as its desired model with probability
\begin{align}\label{eq110}
    q_n &\triangleq\frac{(n_kn/N)^K}{(n_kn/N)^K+(n_k(N-n)/N)^K} \notag \\
    &= \frac{n^K}{n^K+(N-n)^K}
\end{align}
which is independent of the node index $k$ and is denoted by $q_n$.
Then, the second summation in (\ref{eq121}) can be evaluated in a
way that there are $m$ out of $N$ nodes choosing $w^\circ_1$ as
their desired model and the remaining $N-m$ nodes choosing
$w^\circ_0$, which is equal to
\begin{equation}\label{eq125}
    \binom{N}{m}q^m_n(1-q_n)^{N-m}.
\end{equation}
Note that the probability in (\ref{eq125}) depends on $g_{i-1}$ only
through its sum, which is equal to $n$. Therefore, the transition
probability $p_{n,m}$ in (\ref{eq121}) has the same form as
(\ref{eq125}). To evaluate the spectral radius of $Q$ from
(\ref{eq52}), we need the value of $p_{n,0}+p_{n,N}$ (i.e., the
$n$th entry of $b+c$), which is given by:
\begin{equation}\label{eq51}
    p_{n,0}+p_{n,N} = \frac{n^{NK}+(N-n)^{NK}}{(n^K+(N-n)^K)^N}.
\end{equation}
The following result establishes a monotonicity property of the sum
in (\ref{eq51}).
\begin{lem}\label{lem1}
Let $f(x)$ be a function of the form
\begin{equation}
    f(x) = \frac{a^{Nx}+b^{Nx}}{(a^x+b^x)^N}
\end{equation}
for some positive scalars $\{a,b,N\}$ with $N>1$. Then, $f(x)$ is a
non-decreasing function, i.e.,
\begin{equation}
    f'(x) \geq 0
\end{equation}
with equality if, and only if, $a=b$.
\end{lem}
\begin{proof}
The proof follows from evaluating $f'(x)$.
\end{proof}
Since the spectral radius of $Q$ depends on the value of $K$ in
(\ref{eq17}), we will index the quantities with the parameter $K$.
For example, we denote the spectral radius of $Q$ by $\rho[Q(K)]$.
The following result relates the convergence rate of the
decision-making process to the parameter $K$.
\begin{thm}\label{thm4}
The spectral radius $\rho[Q(K)]$ is a strictly decreasing function
of $K$ for $N>2$, i.e.,
\begin{equation}\label{eq128}
    \rho[Q(K+1)] < \rho[Q(K)].
\end{equation}
\end{thm}
\begin{proof}
From (\ref{eq52}), the spectral radius $\rho[Q(K)]$ is given by:
\begin{equation}
    \rho[Q(K)] = 1 -
    \sum_{n=1}^{N-1}y_{Q,n}[p_{n,0}(K)+p_{n,N}(K)]
\end{equation}
where $y_{Q,n}$ is the $n$th entry of $y_Q$ and the sum inside the
brackets is shown in (\ref{eq51}). From Lemma \ref{lem1}, we have
that
\begin{equation}\label{eq83}
    p_{n,0}(K+1)+p_{n,N}(K+1) \geq p_{n,0}(K)+p_{n,N}(K)
\end{equation}
with equality if, and only if, $n=N/2$. Therefore, if $N> 2$, there
exists $n\in{1,2,\cdots,N-1}$ such that strict inequality holds in
(\ref{eq83}). Moreover, since the matrix $Q$ is primitive, the
$\{y_{Q,n}\}$ are positive and we arrive at (\ref{eq128}).
\end{proof}
We therefore conclude that to improve the convergence rate of the
decision-making process, the nodes should use larger values of $K$.
Nevertheless, it may not be beneficial for the network to seek fast
convergence during the decision making process because the network
(e.g., a fish school) may converge to a bad model (e.g., a food
source of poor quality). There exists a trade-off between
exploration and exploitation, as in the case of multi-armed bandit
problem \cite{Gittins89}. Such trade-off can be taken into account
by introducing some weighting scalar $\beta_{k}(i-1)$ that measures
the quality of the desired model of node $k$ at time $i-1$ relative
to the other model. The higher values of $\beta_{k}(i-1)$, the
better the quality of the model and the higher probability that node
$k$ will maintain its desired model. Therefore, node $k$ adjusts the
probability $q_{k,i-1}$ from (\ref{eq57}) to
\begin{equation}
    q_{k,i-1} = \frac{\left[\beta_k(i-1)n^g_k(i-1)\right]^K}
    {\left[\beta_k(i-1)n^g_k(i-1)\right]^K+
    [n_k-n^g_k(i-1)]^K}.
\end{equation}

\subsection{Convergence Rate of Diffusion Adaptation}\label{sec4}
Using the arguments in Section VI, we assume in the following that
the nodes have achieved agreement on the desired model, say,
$w^\circ_q$ as in (\ref{eq30}). We know from the proof of Theorem
\ref{thm2} (see Appendix \ref{appC}) that a modified diffusion
network is equivalent to a network with a mixture of informed and
uninformed nodes, as studied in \cite{Tu13a}. That is, nodes whose
observed model is identical to its desired model ($f(l)=q$) are
informed; otherwise they are uninformed. The convergence rate of the
learning process specifies the rate at which the mean-square error
converges to steady-state. Using the results of \cite{Tu13a}, we can
deduce that the convergence rate, denoted by $r$, of the modified
diffusion strategy (\ref{eq6})-(\ref{eq15}) is given by:
\begin{equation}\label{eq85}
    r = \left[\rho(\mathcal{B})\right]^2
\end{equation}
where $\mathcal{B}$ is defined in Table \ref{tab1}. Note that the
value of $r$ depends on the combination matrix $A$. Under
Assumptions \ref{ass1}-\ref{ass2}, it was shown that the convergence
rate is bounded by \cite{Tu13a}:
\begin{equation}\label{eq86}
    \left(1-\mu\lambda_{\min}(R_u)\right)^2\leq r< 1.
\end{equation}
To improve the convergence rate, it is desirable for the nodes to
select their combination weights so that the network has lower value
of $r$. It was shown in \cite{Tu13a} that for any connected network,
the convergence rate (\ref{eq85}) can achieve the lower bound in
(\ref{eq86}) (namely, the network is able to converge to
steady-state at the fastest rate) by selecting the combination
matrix $A$ according to the following rules:
\begin{enumerate}
\item If there are informed nodes (i.e., nodes with positive step-sizes)
in the neighborhood of node $k$, then it will assign positive
combination weights to those nodes only.
\item Otherwise, node $k$ will assign positive combination
weights to neighbors that are closer (with shorter path) to informed
nodes.
\end{enumerate}
However, there are two issues with this construction. First, it is
difficult to construct the weights in a distributed manner because
rule 2) requires spatial distribution of informed nodes. Second, the
constructed combination matrix is not primitive (i.e., Assumption
\ref{ass3} does not hold) because there are no links from uninformed
nodes to informed nodes. Therefore, Theorem \ref{thm2} would not
apply here. In the following, we first propose a way to select
combination weights that approximate rule 2) and then show that the
approximate weights ensure mean convergence.

Let $\mathcal{N}_{k}^f$ denote the set of nodes that are in the
neighborhood of $k$ and whose observed model is the same as the
desired model $w^\circ_q$ (i.e., they are informed neighbors)
\begin{equation}
    \mathcal{N}_{k}^f = \{l\;|\;l\in\mathcal{N}_k,
    f(l)=q\}.
\end{equation}
Also, let $n^f_k$ denote the number of nodes in the set
$\mathcal{N}_{k}^f$. The selection of combination weights is
specified based on three types of nodes: informed nodes ($f(k)=q$),
uninformed nodes with informed neighbors ($f(k)\neq q$ and
$n^f_k\neq 0$), and uninformed nodes without informed neighbors
($f(k)\neq q$ and $n^f_k = 0$). The first two types correspond to
rule 1) and their weights can satisfy rule 1) by setting
\begin{equation}\label{eq95}
    a_{l,k} = \begin{cases}
    1/n^f_k, &\text{if $l\in\mathcal{N}^f_k$}\\
    0, &\text{otherwise}
    \end{cases}.
\end{equation}
That is, node $k$ places uniform weights on the informed neighbors
and zero weights on the others. The last type of nodes corresponds
to rule 2). Since these nodes do not know the distribution of
informed nodes, a convenient choice for the approximate weights they
can select is for them to place zero weights on themselves and
uniform weights on the others, i.e.,
\begin{equation}\label{eq96}
    a_{l,k} = \begin{cases}
    1/(n_k-1), &\text{if $l\in\mathcal{N}_k$ and $l\neq k$}\\
    0, &\text{otherwise}
    \end{cases}.
\end{equation}
Note that the weights from (\ref{eq95})-(\ref{eq96}) can be set in a
fully distributed manner and in real-time. To show the mean
convergence of the modified diffusion strategy using the combination
matrix $A$ constructed from (\ref{eq95})-(\ref{eq96}), we resort to
Theorem 1 from \cite{Tu13a}. It states that the strategy converges
in the mean for sufficiently small step-sizes if for any node $k$,
there exists an informed node $l$ and an integer power $j$ such that
\begin{equation}\label{eq97}
    [A^j]_{l,k} > 0.
\end{equation}
Condition (\ref{eq97}) is clearly satisfied for the first two types
of nodes. For any node belonging to the last type, since the network
is connected and from (\ref{eq96}), there exists a path with nonzero
weight from a node of the second type (uninformed with informed
neighbors) to itself. In addition, there exist direct links from
informed nodes to the nodes of the second type, condition
(\ref{eq97}) is also satisfied. This implies that the modified
diffusion strategy using the combination weights from
(\ref{eq95})-(\ref{eq96}) converges in the mean.

\section{Simulation Results}
We consider a network with 40 nodes randomly connected. The model
vectors are set to $w^\circ_0=[5;-5;5;5]$ and $w^\circ_1 =
[5;5;-5;5]$ (i.e. $M=4$). Assume that the first 20 nodes (nodes 1
through 20) observe data originating from model $w^\circ_0$, while
the remaining nodes observe data originating from model $w^\circ_1$.
The regression covariance matrix $R_u$ is diagonal with each
diagonal entry generated uniformly from $[1,2]$. The noise variance
at each node is generated uniformly from $[-35,-5]$ dB. The
step-sizes are set to $\mu=0.005$, $\nu=0.05$, and $\alpha=0.95$.
The threshold $\eta$ is set to $\eta =1$. The network employs the
decision-making process with $K=4$ in (\ref{eq57}) and the uniform
combination rule: $a_{l,k} =1/n_k$ if $l\in\mathcal{N}_k$.

In Fig. \ref{Fig_5}, we illustrate the network mean-square deviation
(MSD) with respect to the two model vectors over time, i.e.,
\begin{equation}
    \text{MSD}_q(i) = \frac{1}{N}\sum_{k=1}^N\mathbb{E}
    \|w^\circ_q-\boldsymbol{w}_{k,i}\|^2
\end{equation}
for $q=0$ and $q=1$. We compare the conventional ATC diffusion
strategy (\ref{eq49})-(\ref{eq1}) and the modified ATC diffusion
strategy (\ref{eq6})-(\ref{eq15}) with decision-making. We observe
the bifurcation in MSD curves of the modified ATC diffusion
strategy. Specifically, the MSD curve relative to the model
$w^\circ_0$ converges to 23 dB, while the MSD relative to
$w^\circ_1$ converges to -50 dB. This illustrates that the nodes
using the modified ATC diffusion are able to agree on a model and to
converge to steady-state (to model $w^\circ_1$ in this case). We
also show in Fig. \ref{Fig_6} the evolution of the beliefs
$\{b_{k,i}(l)\}$ for a particular node using the update rule
(\ref{eq16}). The node has two neighbors observing data that
originate from the same model and two neighbors observing data from
a different model. We observe that, at the initial stage of
adaptation, all beliefs increase. Nevertheless, as time evolves, the
node is able to differentiate between the two models and the beliefs
for the latter two neighbors decrease. Note that the belief
converges to one if a node has the same observed model; otherwise,
it converges to zero. This indicates that the classification scheme
successfully identifies the observed models of neighboring nodes. On
the other hand, for the conventional diffusion strategy, the nodes
also converge because the MSD curves in Fig. \ref{Fig_5} remain
flat. However, the MSD values are large (about 18 dB). This implies
that the nodes converge to a common vector that does not coincide
with either of the model vectors.

\begin{figure}
\centering
\includegraphics[width=20em]{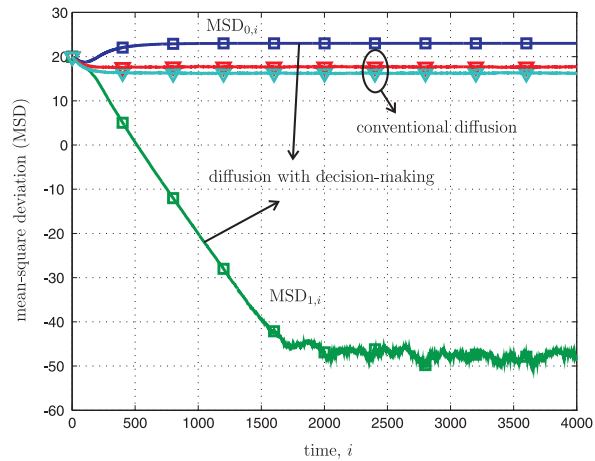}
\caption{Transient network MSD over a network using the conventional
diffusion strategy (\ref{eq49})-(\ref{eq1}) and using the modified
diffusion strategy (\ref{eq6})-(\ref{eq15}). The network with
decision-making converges to the model $w^\circ_1$ while the network
without decision making converges to a vector that is not identical
to either of the model vectors.} \label{Fig_5}
\end{figure}

\begin{figure}
\centering
\includegraphics[width=20em]{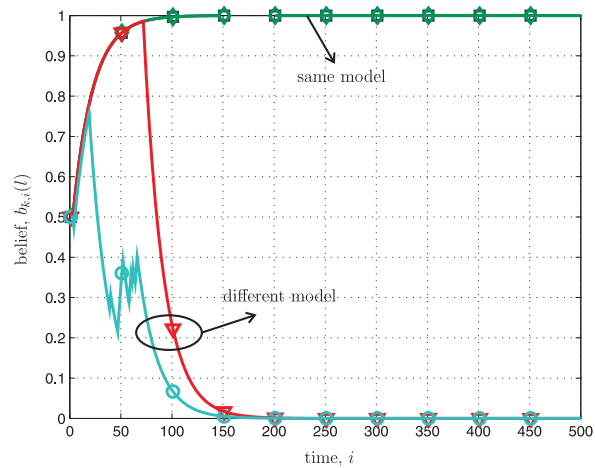}
\caption{Evolution of beliefs using (\ref{eq16}) at a particular
node. The node has four neighbors; two of them collect data from the
same model while the other two collect data from a different model.}
\label{Fig_6}
\end{figure}

We also show the dependence of the convergence rate on the parameter
$K$. We compare two modified diffusion strategies using
decision-making with $K=1$ and $K=4$ in (\ref{eq57}). The network
MSD curves for these two strategies are shown in Fig. \ref{Fig_8}.
We observe that the MSD curves relative to the model $w^\circ_1$
decrease at the same rate and converge to the same steady-state
value. However, there is about 75 shift in time between these
curves: the MSD$_{1,i}$ with $K=4$ is 75 time steps ahead of the MSD
curve with $K=1$. As the analytical result revealed, the
decision-making processes adopting larger values of parameter $K$
achieve agreement at faster rate. We also consider the effect of the
combination weights on the convergence rate of the adaptation
strategies. Figure \ref{Fig_9} illustrates the modified diffusion
strategies with different combination weights: one with the uniform
combination rule and the other one with the combination rule in
(\ref{eq95})-(\ref{eq96}). We observe that the diffusion strategy
using the proposed rule converges at faster rate with some
degradation in steady-state MSD. Note that the trade-off between
convergence rate and MSD is also indicated in \cite{Tu13a}.

\begin{figure}
\centering
\includegraphics[width=20em]{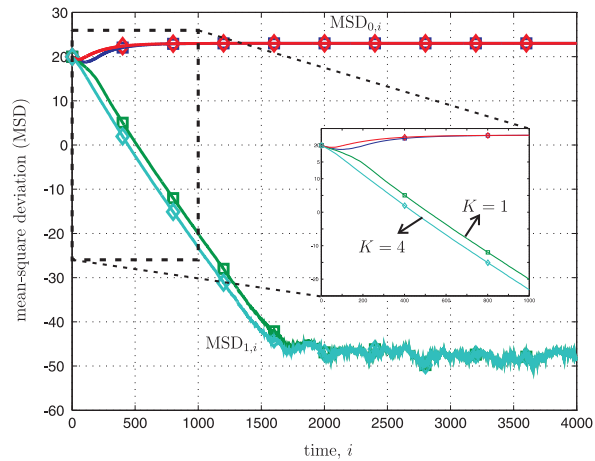}
\caption{Transient network MSD over the modified diffusion
strategies (\ref{eq6})-(\ref{eq15}) with decision-making process for
$K=1$ and $K=4$ in (\ref{eq57}).} \label{Fig_8}
\end{figure}

\begin{figure}
\centering
\includegraphics[width=20em]{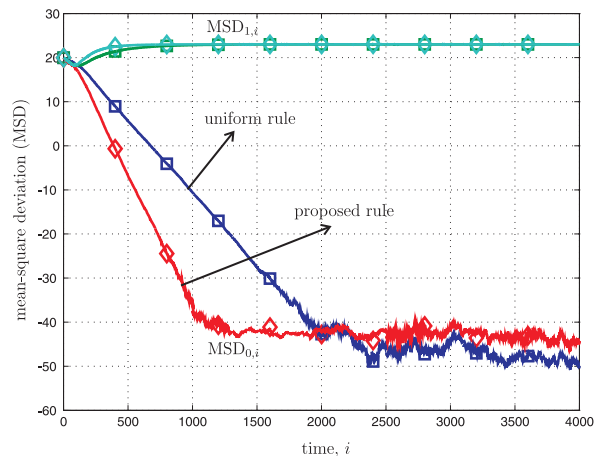}
\caption{Transient network MSD over the modified diffusion strategy
(\ref{eq6})-(\ref{eq15}) using the uniform combination rule and the
proposed rule (\ref{eq95})-(\ref{eq96}).} \label{Fig_9}
\end{figure}

We apply the results of this paper to model the fish schooling
behavior in the presence of two food sources (located at $w^\circ_0$
and $w^\circ_1$). It is observed in nature that fish move in a
harmonious manner so that they align their motion and keep a safe
distance from each other
\cite{Jadbabaie03,Gazi04,Olfati06,Lorenzo11,Zavlanos11}. We apply
the motion control mechanism from \cite{Tu11a} to model mobile
agents. Let $x_{k,i}$ denote the location vector of node $k$ at time
$i$. Every node $k$ adjusts its location vector according to the
rule:
\begin{equation}\label{eq105}
    x_{k,i+1} = x_{k,i}+\Delta t \cdot v_{k,i+1}
\end{equation}
where $\Delta t$ is a positive time step and $v_{k,i+1}$ is the
velocity vector at node $k$, which is set according to the rule:
\begin{equation}\label{eq106}
    v_{k,i+1} = \lambda \frac{w_{k,i}-x_{k,i}}{\|w_{k,i}-x_{k,i}\|}
    +\beta \sum_{l\in\mathcal{N}_k}c_{l,k}v_{k,i}+\gamma \delta_{k,i}
\end{equation}
where $\{\lambda,\beta,\gamma\}$ are nonnegative scalars and
$\delta_{k,i}$ helps the nodes keep a certain distance $d_s$ to each
other and is given by
\begin{equation*}
    \delta_{k,i} = \frac{1}{n_k-1}\sum_{l\in\mathcal{N}_k\setminus \{k\}}
    \left(\|x_{l,i}-x_{k,i}\|-d_s\right)\frac{x_{l,i}-x_{k,i}}{\|x_{l,i}-x_{k,i}\|}.
\end{equation*}
The nodes employ the diffusion strategy to estimate the location of
food sources. This is achieved as follows. We showed in \cite{Tu11a}
that the distance, $d^\circ_k(i)$, between the target located at
$w^\circ$ and a node $k$ located at $x_{k,i}$ can be expressed as
the inner product (see Fig. \ref{Fig_15}):
$d^\circ_k(i)=u^\circ_{k,i}(w^{\circ}-x_{k,i})$ where
$u^\circ_{k,i}$ denotes the unit direction vector pointing to
$w^\circ$ from $x_{k,i}$. However, the nodes observe a noisy
distance $d_k(i)$ and a noisy direction $u_{k,i}$ to the target,
which can be related to $w^\circ$ as follows (the same form as
(\ref{eq2})):
\begin{align}
    \hat{d}_k(i) \triangleq d_k(i)+u_{k,i}x_{k,i}
    = u_{k,i}w^\circ+v_{k}(i) \label{eq14}
\end{align}
where $v_k(i)$ is the scalar noise term and its variance is
proportional to the distance of node $k$ to the target, i.e.,
\begin{equation}
    \sigma^2_{v,k,i}=\kappa\|w^\circ-\boldsymbol{x}_{k,i}\|^2
\end{equation}
with $\kappa=0.01$. In simulation, there are two targets located at
$w^\circ_0=[10,10]$ and $w^\circ_1=[-10,10]$. The nodes then apply
Algorithm in Section \ref{sec5} to achieve agreement on a desired
target. The simulation results are illustrated in Fig. \ref{Fig_14}.
The parameters used in (\ref{eq105})-(\ref{eq106}) are set to
$(\Delta t,\lambda,\beta,\gamma,d_s)=(0.1,0.3,0.7,1,3)$. Initially,
there are $40$ nodes uniformly distributed in a $20\times 20$ square
area around the origin. There are $20$ nodes collecting data that
originate from target $w^\circ_0$ and the remaining $20$ nodes
collecting data arising from the other target $w^\circ_1$. In Fig.
\ref{Fig_14}, nodes that would like to move towards $w^\circ_0$ are
shown as blue dots and nodes that would like to move towards
$w^\circ_1$ are shown as red circles. We observe that the node
achieve agreement on a desired target and get to the target (at
$w^\circ_1=[40,-40]$ in this case).

\begin{figure}
\centering
\includegraphics[width=16em]{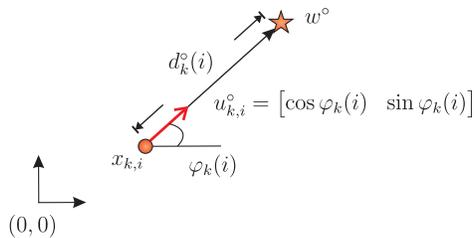}
\caption{Distance and direction of the target $w^\circ$ from node
$k$ at location $x_k$. The unit direction vector $u_k^\circ$ points
towards $w^\circ$.} \label{Fig_15}
\end{figure}

\begin{figure}
\centering
\includegraphics[width=24em]{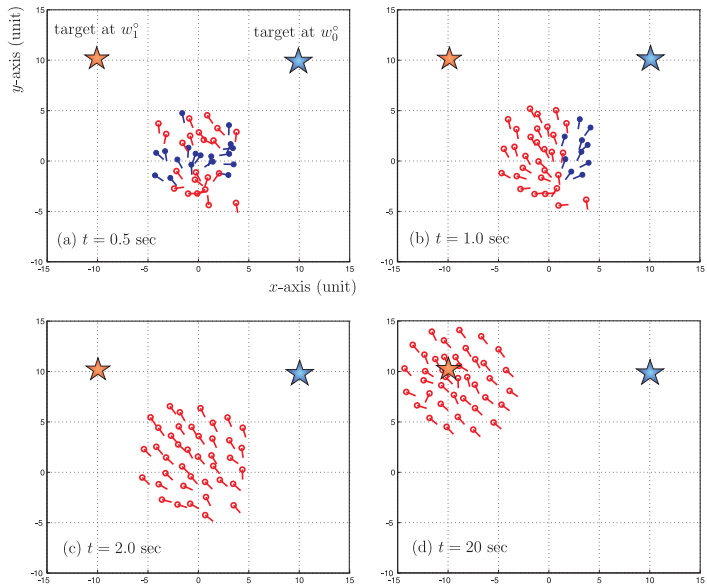}
\caption{Maneuver of fish schools with two food sources over time
(a) $t=0.5$ (b) $t=1$ (c) $t=2$ and (d) $t=20$ sec.} \label{Fig_14}
\end{figure}

\section{Concluding Remarks}
In the presence of distinct models observed by the nodes in a
network, conventional distributed estimation strategies will lead to
biased solutions. In this paper, we proposed a modified strategy to
address this issue. To do so, we allow the nodes to exchange not
only intermediate estimates, but also previous estimates. We also
developed a classification scheme and a decision-making procedure
for the nodes to identify the underlying models that generate data
and to achieve agreement among the nodes on the desired objective.
It is useful to comment on some features of the proposed framework.

We focused in this work on the case where nodes need to choose
between two models. Extension of the techniques to multiple models
require additional analysis. The case of two models is not a serious
limitation especially since many hypothesis testing problems tend to
be formulated as deciding between two choices. In addition, it is
natural to expect that convergence of the decision process will
occur towards one model or the other in a probabilistic manner since
the outcome is influenced by the fraction of nodes that sense data
from one model or another. Interestingly, though, the
decision-making process and the estimation task are largely
independent of each other. This is because there are two tasks that
the nodes need to accomplish. First, they need to decide which of
the two models to follow and, second, they need to estimate the
model. To solve the first task, agents do not need to know the exact
model values. An arbitrary node $k$ only needs to know whether a
neighboring node $l$ is observing data from the same model or from a
different model regardless of the model values. This property
enables the initial decision process to converge faster and to be
largely independent of the estimation task.

{\appendices
\section{Proof of Theorem \ref{thm2}}\label{appC}
Without loss of generality, let $w^\circ_0$ be the desired model for
the network (i.e., $q=0$ in (\ref{eq30})) and assume there are $N_0$
nodes with indices $\{1,2,\ldots,N_0\}$ observing data arising from
the model $w^\circ_0$, while the remaining $N-N_0$ nodes observe
data arising from model $w^\circ_1$. Then, we obtain from
(\ref{eq29}), (\ref{eq24}), and (\ref{eq25}) that
\begin{align}
    \tilde{z}^\circ_k &=
    \begin{cases}
    0, &\text{if $k\leq N_0$}\\
    w^\circ_0-w^\circ_1, &\text{if $k>N_0$}
    \end{cases}\\
    a^{(1)}_{l,k} &= 0 \text{ if $l>N_0$} \quad\text{and}\quad
    a^{(2)}_{l,k} = 0 \text{ if $l\leq N_0$}.
\end{align}
Since the matrix $\mathcal{M}\mathcal{R}$ is block diagonal, we
conclude that
\begin{equation}
    y=0\quad\text{and}\quad
    \mathcal{B} = \mathcal{A}^T(I_{NM}-\mathcal{M}_e\mathcal{R})
\end{equation}
where $\mathcal{M}_e$ is an $N\times N$ block diagonal matrix of the
form
\begin{equation}\label{eq84}
    \mathcal{M}_e \triangleq \text{diag}
    \{\mu_1I_M,\cdots,\mu_{N_0}I_M,0,\cdots,0\}.
\end{equation}
That is, its mean recursion in (\ref{eq5}) is equivalent to the mean
recursion of a network running the traditional diffusion strategy
(\ref{eq49})-(\ref{eq1}) with $N_0$ nodes (nodes 1 to $N_0$) using
positive step-sizes and $N-N_0$ nodes (nodes $N_0+1$ to $N$) having
zero step-sizes. Then, according to Theorem 1 of \cite{Tu13a} and
under the assumption that the matrix $A$ is primitive, if the
step-sizes $\{\mu_1,\mu_2,\cdots,\mu_{N_0}\}$ are set to satisfy
(\ref{eq26}), then the spectral radius of $\mathcal{B}$ will be
strictly less than one.

\section{Proof of Theorem \ref{thm1}} \label{appD}
For a given vector $\boldsymbol{g}_{i-1}$, we denote by $\chi_{i-1}$
the number of nodes whose desired model is $w^\circ_1$ at time
$i-1$, i.e.,
\begin{equation}\label{eq103}
    \chi_{i-1} \triangleq \sum_{k=1}^N \boldsymbol{g}_{i-1}(k).
\end{equation}
From (\ref{eq98})-(\ref{eq57}), the vector $\boldsymbol{g}_{i}$
depends only on $\boldsymbol{g}_{i-1}$. Thus, the value of
$\chi_{i}$ depends only on $\chi_{i-1}$. Therefore, the evolution of
$\chi_i$ forms a Markov chain with $N+1$ states corresponding to the
values $\{0,1,2,\ldots,N\}$ for $\chi_i$. To compute the transition
probability, $p_{n,m}$, from state $\chi_{i-1}=n$ to state
$\chi_i=m$, let us denote by $\mathcal{G}_{n}$ the set of vectors
$g=\{g(1),g(2),\cdots,g(N)\}$ whose entries are either 1 or 0 and
add up to $n$, i.e.,
\begin{equation}\label{eq120}
    \mathcal{G}_{n} = \left\{g \;|\;
    \sum_{k=1}^N g(k)=n\right\}.
\end{equation}
Then, the $p_{n,m}$ can be written as:
\begin{align}\label{eq121}
    p_{n,m}=\sum_{\boldsymbol{g}_{i-1}\in\mathcal{G}_{n}}
    \Pr(\boldsymbol{g}_{i-1})\sum_{\boldsymbol{g}_i\in\mathcal{G}_{m}}
    \prod_{l=1}^N
    \Pr(\boldsymbol{g}_{i}(l)\;|\; \boldsymbol{g}_{i-1}(l))
\end{align}
where $\Pr(\boldsymbol{g}_{i-1})$ is \emph{a priori} probability and
where the probability $\Pr(\boldsymbol{g}_{i}(l)\;|\;
\boldsymbol{g}_{i-1}(l))$ is determined by (\ref{eq57}). Note that
for a static network, the transition probability $p_{n,m}$ is
independent of $i$, i.e., the Markov chain is \emph{homogeneous}
\cite{Lawler06}.

Now we assume that $\chi_{i-1}=n\neq0,N$. Since the network is
connected, for any $\boldsymbol{g}_{i-1}\in\mathcal{G}_{n}$ at least
one node (say, node $k$) has desired model $w^\circ_1$ and has a
neighbor with distinct desired model $w_0^\circ$ so that
$n^g_{k}(i-1) < n_{k}$ and $1-q_{k,i-1} > 0$ from (\ref{eq57}).
Since $q_{l,i-1}>0$ for all $l$, we obtain from (\ref{eq121}) that
\begin{align}
    p_{n,n-1}
    &\geq \sum_{\boldsymbol{g}_{i-1}\in\mathcal{G}_{n}}\Pr(\boldsymbol{g}_{i-1})
    (1-q_{k,i-1}) \prod_{l\neq k}q_{l,i-1}
    >0\notag \\
    p_{n,n} &>0\quad\text{and}\quad p_{n,n+1}>0 \label{eq123}
\end{align}
for $n\neq0,N$. When $n=0$ or $n=N$, we have that
$p_{0,0}=p_{N,N}=1$. This indicates that the Markov chain has two
absorbing states: $\chi_i=0$ (or,
$\boldsymbol{g}_i(1)=\boldsymbol{g}_i(2)=\cdots=\boldsymbol{g}_i(N)=0$)
and $\chi_i=N$ (or,
$\boldsymbol{g}_i(1)=\boldsymbol{g}_i(2)=\cdots=\boldsymbol{g}_i(N)=1$),
and for any state $\chi_i$ different from $0$ and $N$, there is a
nonzero probability traveling from an arbitrary state $\chi_i$ to
state $0$ and state $N$. Therefore, no matter which state the Markov
chain starts from, it converges to state $0$ or state $N$
\cite[p.26]{Lawler06}, i.e., all nodes reach agreement on the
desired model.

\section{Proof of Lemma \ref{lem5}}\label{appA}
Let $C_1$ denote the far-field condition: $\|z^\circ_k-w_{k,i-1}\|
\gg 1$. We obtain from Assumption \ref{ass5} and (\ref{eq75}) that
\begin{align}
    \Pr(\|\hat{\boldsymbol{h}}_{k,i}\| > \eta\;|\; C_1)
    &\overset{(a)}{\geq} \Pr(\|\bar{h}_{k,i}\|-\|\boldsymbol{n}_{k,i}\| > \eta\;|\;C_1)\notag \\
    &= 1- \Pr(\|\boldsymbol{n}_{k,i}\| \geq \|\bar{h}_{k,i}\|-\eta\;|\;C_1)\notag \\
    &\overset{(b)}{\geq} 1- \frac{\mathbb{E}\|\boldsymbol{n}_{k,i}\|^2}
    {(\|\bar{h}_{k,i}\|-\eta)^2}\notag \\
    &\overset{(c)}{=} 1- \frac{\nu[\tau\|\bar{h}_{k,i}\|^2+\sigma^2_{v,k}\text{Tr}(R_u)]}
    {2(\|\bar{h}_{k,i}\|-\eta)^2} \label{eq41}
\end{align}
where step (a) follows from the triangle inequality of norms and
(\ref{eq60}), step (b) is by the Markov inequality (\ref{eq61}) and
Assumption \ref{ass5}, and step (c) is by (\ref{eq40}). Moreover,
under conditions (\ref{eq72}) and $C_1$, we can ignore the term
$\eta$ in the denominator of (\ref{eq41}). In addition, from
condition $C_1$ and (\ref{eq76}), and since the variance
$\sigma^2_{v,k}$ is generally small, we may ignore the term
$\nu\sigma^2_{v,k}\text{Tr}(R_u)$ in (\ref{eq41}) and obtain
(\ref{eq62}). Similar arguments apply to (\ref{eq74}).

\section{Proof of Lemma \ref{lem6}}\label{appB}
Under condition (\ref{eq72}) and from (\ref{eq16}), the probability
$P_d$ in (\ref{eq100}) becomes
\begin{align}
    P_d
    &= \Pr(\|\hat{\boldsymbol{h}}_{k,i}\| > \eta,\;
    \|\hat{\boldsymbol{h}}_{l,i}\| > \eta,\;
    \hat{\boldsymbol{h}}_{k,i}^T\hat{\boldsymbol{h}}_{l,i}>0
    \;|\;z^\circ_k=z^\circ_l)\notag \\
    &\approx \Pr(\hat{\boldsymbol{h}}_{k,i}^T\hat{\boldsymbol{h}}_{l,i}>0
    \;|\;z^\circ_k=z^\circ_l)\label{eq45}
\end{align}
where we used the fact that $\hat{\boldsymbol{h}}_{k,i}$ and
$\hat{\boldsymbol{h}}_{l,i}$ are independent, as well as the result
of Lemma \ref{lem5} which ensures that
$\|\hat{\boldsymbol{h}}_{k,i}\|>\eta$ with high probability
(likewise, for the norm of $\hat{\boldsymbol{h}}_{l,i}$). Note that
the event $\hat{\boldsymbol{h}}_{k,i}^T\hat{\boldsymbol{h}}_{l,i}>0$
is equivalent to the fact that the angle between these two vectors
is less than $\pi/2$. Let $\boldsymbol{\theta}_{k,i}$ denote the
angle between the vectors $\bar{h}_{k,i}$ and
$\hat{\boldsymbol{h}}_{k,i}$ due to the noise $\boldsymbol{n}_{k,i}$
(see Fig. \ref{Fig_10}(a)). The value of $\boldsymbol{\theta}_{k,i}$
is positive if the vector $\hat{\boldsymbol{h}}_{k,i}$ rotates
counter-clockwise relative to $\bar{h}_{k,i}$; otherwise, its value
is negative. Then, we have that the angle
$\boldsymbol{\theta}_{k,i}$ is upper bounded by (see Fig.
\ref{Fig_10}(a)):
\begin{equation}\label{eq42}
    |\boldsymbol{\theta}_{k,i}|\leq
    \sin^{-1}\left(\frac{\|\boldsymbol{n}_{k,i}\|}
    {\|\bar{h}_{k,i}\|}\right)
    \approx
    \frac{\|\boldsymbol{n}_{k,i}\|}
    {\|\bar{h}_{k,i}\|}
\end{equation}
\noindent That is, the maximum value of $\boldsymbol{\theta}_{k,i}$
occurs when the vectors $\hat{\boldsymbol{h}}_{k,i}$ and
$\boldsymbol{n}_{k,i}$ are perpendicular. The approximation in
(\ref{eq42}) is from (\ref{eq72}) and (\ref{eq40}) so that it holds
that
\begin{equation}\label{eq66}
    \frac{\mathbb{E}\|\boldsymbol{n}_{k,i}\|^2}{\|\bar{h}_{k,i}\|^2}
    \leq \frac{\nu[\tau\|\bar{h}_{k,i}\|^2+\sigma^2_{v,k}\text{Tr}(R_u)]}
    {2\|\bar{h}_{k,i}\|^2}
    \approx \frac{\nu \tau}{2}.
\end{equation}
Since all nodes start from the same initial estimate (i.e.,
$w_{k,-1}=0$ for all $k$), the estimates
$\{\boldsymbol{w}_{k,i-1}\}$ are close to each other during the
initial stages of adaptation and it is reasonable to assume that
$\|\boldsymbol{w}_{k,i-1}-\boldsymbol{w}_{l,i-1}\| \ll
\|z^\circ_k-\boldsymbol{w}_{k,i-1}\|$. Therefore, we arrive at the
approximation $\bar{h}_{l,i}\approx \bar{h}_{k,i}$ for computing
$P_d$. This implies that the vectors $\hat{\boldsymbol{h}}_{k,i}$
and $\hat{\boldsymbol{h}}_{l,i}$ can be modeled as starting
approximately at the same location $\boldsymbol{w}_{k,i-1}$ but
having deviated by angles $\boldsymbol{\theta}_{k,i}$ and
$\boldsymbol{\theta}_{l,i}$, respectively (see Fig.
\ref{Fig_10}(b)). Therefore, the angle between
$\hat{\boldsymbol{h}}_{k,i}$ and $\hat{\boldsymbol{h}}_{l,i}$ is
equal to $|\boldsymbol{\theta}_{k,i}-\boldsymbol{\theta}_{l,i}|$.
From (\ref{eq45}), we obtain that
\begin{align}
    P_d &\approx \Pr\left(|\boldsymbol{\theta}_{k,i}-\boldsymbol{\theta}_{l,i}| <
    \frac{\pi}{2}\;|\; z^\circ_k=z^\circ_l\right)\notag \\
    &\overset{(a)}{\geq} \Pr\left(|\boldsymbol{\theta}_{k,i}|+|\boldsymbol{\theta}_{l,i}| <
    \frac{\pi}{2}\;|\; z^\circ_k=z^\circ_l\right)\notag \\
    &\overset{(b)}{\geq} \Pr\left(\frac{\|\boldsymbol{n}_{k,i}\|}{\|\bar{h}_{k,i}\|}
    +\frac{\|\boldsymbol{n}_{l,i}\|}{\|\bar{h}_{k,i}\|} <
    \frac{\pi}{2}\right)\notag \\
    &= 1- \Pr(\|\boldsymbol{n}_{k,i}\|+\|\boldsymbol{n}_{l,i}\| \geq
    \pi\|\bar{h}_{k,i}\|/2) \label{eq65}
\end{align}
where step (a) is by the triangle inequality of norms and
(\ref{eq60}) and step (b) is by (\ref{eq42}). To evaluate the
probability in (\ref{eq65}), we resort to the following fact. For
any two random variables $\boldsymbol{x}$ and $\boldsymbol{y}$ and
for any constant $\eta$, it holds from (\ref{eq60}) that
\begin{align}
    \Pr(\boldsymbol{x}+\boldsymbol{y}>\eta)
    \leq \Pr(\boldsymbol{x} >
    \eta/2)+\Pr(\boldsymbol{y}>\eta/2).\label{eq80}
\end{align}
This leads to
\begin{equation}
    P_d \geq 1-\Pr(\|\boldsymbol{n}_{k,i}\| > \pi\|\bar{h}_{k,i}\|/4)
    -\Pr(\|\boldsymbol{n}_{l,i}\| > \pi\|\bar{h}_{k,i}\|/4).
\end{equation}
We then arrive at (\ref{eq68}) because
\begin{align}
    P_d\geq 1- \frac{16(\mathbb{E}\|\boldsymbol{n}_{k,i}\|^2+
    \mathbb{E}\|\boldsymbol{n}_{l,i}\|^2)}{\pi^2\|\bar{h}_{k,i}\|^2}
    \geq 1-\frac{16\nu \tau}{\pi^2}
\end{align}
where we used the Markov inequality (\ref{eq61}) and (\ref{eq66}).
Similar arguments apply to $P_f$ when $z^\circ_k\neq z^\circ_l$ by
noting that the vectors $\bar{h}_{k,i}$ and $\bar{h}_{l,i}$ can
again be modeled as starting approximately at the same location
$w_{k,i-1}$, but pointing towards different directions:
$\bar{h}_{k,i}$ towards $z_k$ and $\bar{h}_{l,i}$ towards $z_l$, and
the angle between these two vectors now assumes a value close to
$\pi$ according to Lemma \ref{lem3}.}

\begin{figure}
\centering
\includegraphics[width=22em]{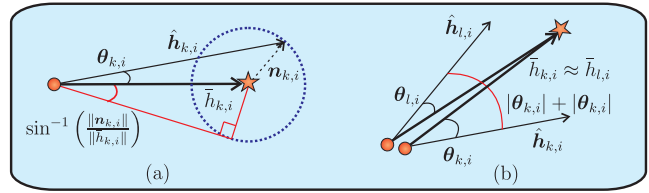}
\caption{Illustration of (a) the angle $\boldsymbol{\theta}_{k,i}$
between $\bar{h}_{k,i}$ and $\hat{\boldsymbol{h}}_{k,i}$ due to the
noise $\boldsymbol{n}_{k,i}$ and (b) the angle between
$\hat{\boldsymbol{h}}_{k,i}$ and $\hat{\boldsymbol{h}}_{l,i}$ when
$z^\circ_k=z^\circ_l$.} \label{Fig_10}
\end{figure}

\bibliographystyle{IEEEtran}
\bibliography{IEEEfull,refs}

\end{document}